\documentclass[aps,prl,onecolumn,superscriptaddress,amsfonts,amsmath,amssymb,floatfix,showpacs]{revtex4-1}
\usepackage{graphicx}
\usepackage{graphics}
\usepackage{amsmath}
\usepackage{float}
\usepackage{subfigure}
\usepackage{amsthm}
\usepackage{longtable}
\newcommand{\tr}{\mathop{\mathrm{tr}}}
\newcommand{\rank}{\mathop{\mathrm{rank}}}
\newtheorem*{theorem}{Theorem}
\newcommand{\spann}{\mathop{\mathrm{span}}}
\newcommand{\diag}{\mathop{\mathrm{diag}}}
\newtheorem{lemma}{Lemma}
\newtheorem{hypothesis}{Hypothesis}

\begin{document}


\title{Experimental Device-Independent Tests of Classical and Quantum Entropy}

\author{Feng Zhu}
\affiliation{Tsinghua National Laboratory for Information Science and Technology, Department of Electronic Engineering, Tsinghua University, Beijing, 100084, P. R. China}
\author{Wei Zhang}
\email{zwei@tsinghua.edu.cn}
\affiliation{Tsinghua National Laboratory for Information Science and Technology, Department of Electronic Engineering, Tsinghua University, Beijing, 100084, P. R. China}
\author{Sijing Chen}
\affiliation{State Key Laboratory of Functional Materials for Informatics, Shanghai Institute of Microsystem and Information Technology (SIMIT), Chinese Academy of Sciences, Shanghai 200050, China}
\author{Lixing You}
\affiliation{State Key Laboratory of Functional Materials for Informatics, Shanghai Institute of Microsystem and Information Technology (SIMIT), Chinese Academy of Sciences, Shanghai 200050, China}
\author{Zhen Wang}
\affiliation{State Key Laboratory of Functional Materials for Informatics, Shanghai Institute of Microsystem and Information Technology (SIMIT), Chinese Academy of Sciences, Shanghai 200050, China}
\author{Yidong Huang}
\affiliation{Tsinghua National Laboratory for Information Science and Technology, Department of Electronic Engineering, Tsinghua University, Beijing, 100084, P. R. China}

\date{\today}


\begin{abstract}

In this paper, we report an experiment about the device-independent tests of classical and quantum entropy based on a recent proposal [Phys. Rev. Lett. $\bf{115}$, 110501 (2015)], in which the states are encoded on the polarization of a biphoton system and measured by the state tomography technology. We also theoretically obtained the minimal quantum entropy for three widely used linear dimension witnesses. The experimental results agree well with the theoretical analysis, demonstrating that lower entropy is needed in quantum systems than that in classical systems under given values of the dimension witness.

\pacs{03.67.Mn,03.67.Dd,03.67.Hk,42.50.Xa}

\end{abstract}


\maketitle


\textit{Introduction}{\textemdash}The device-independent quantum information processing is attractive and developing rapidly. Since the imperfection of practical devices will reduce the security of the quantum key distribution, the device-independent quantum key distribution was proposed against the collective attacks from the eavesdroppers\cite{acin2007prl}. It is independent of the internal working of the devices used in the implementation. The security is guaranteed from the observed data without any reference on the states and measurements.

Tests of resources in quantum information are also proposed in the device-independent manner, in which the source and the detector in the prepare-and-measure scenario are regarded as ``black boxes''. For example, the entanglement\cite{Elkert1991prl} is a basic resource in quantum communication and quantum computation. Tests of the entanglement in the device-independent manner have been theoretically analyzed\cite{liang2011prl,Rabelo2011prl,liang2013prl} and experimentally demonstrated\cite{Barreiro2013natphy}. The dimension\cite{brunner2008prl} is another important resource for the system used in the quantum information processing. It can also be tested in the device-independent way\cite{gallego2010prl,bowles2014prl,pawlowski2011pra} and has been demonstrated experimentally\cite{hendrych2014natphy,ahrens2012natphy,ambrosio2014prl,ahrens2014prl}.



Entropy is an important fundamental resource which reveals the amount of information in the communication tasks\cite{holevo,palowski2009nat}. Device-independent tests of entropy were proposed recently\cite{chaves2015prl}. It is realized by constructing two entropy witnesses. The first one is based on the causal inference networks\cite{pearl2009}, in which the facets of the entropic cone can be characterized\cite{chaves2015nc,chaves2013arxiv,chaves2013ieee} associating a directed acyclic graph. It is a general method and valid for systems with arbitrary finite dimensions. However, it has an important drawback that it cannot discriminate the classical case from the quantum case, since the lower bounds of the classical and quantum entropy calculated by this way are the same. The other way is based on the convex optimization techniques, which can reveal the difference between the classical entropy and the quantum entropy\cite{chaves2015prl}. Utilizing the value of the dimension witness, the minimal classical entropy can be explicitly derived. An upper bound of the minimal quantum entropy can also be obtained using 4-dimensional systems. Whether it is exactly the minimal quantum entropy has not been investigated, since it is not clear that whether higher dimensional systems can be used to reduce the quantum entropy.




In this Letter, we theoretically investigate the minimal quantum entropy in systems with arbitrary dimension for any linear dimension witness, showing that it cannot be reduced by using higher dimensional systems and it is lower than the minimal classical entropy under the given value of the dimension witness. The classical and quantum entropy are tested experimentally, demonstrating their significant difference.




\begin{figure}[!htb]
\includegraphics[width=2.9 in]
{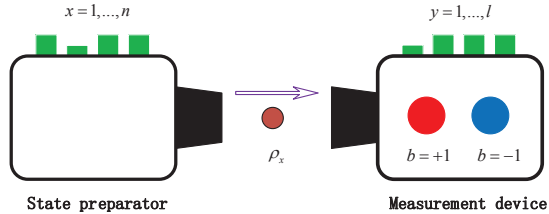}
\caption{Prepare-and-measure scenario.\label{g-1}}
\end{figure}

\begin{figure*}[!htbp]
\centering
\includegraphics[width=7.3 in]
{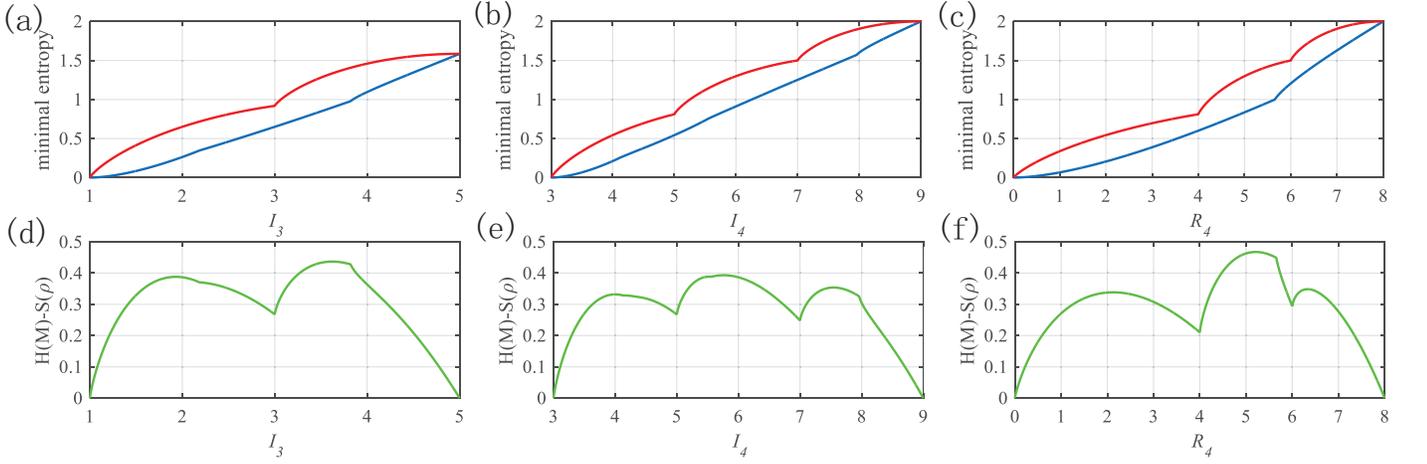}
\caption{The minimal classical(red) and quantum(blue) entropy under given values of different dimension witnesses. (a), (b) and (c) are the results for the dimension witnesses of $I_3$, $I_4$ and $R_4$. (d), (e) and (f) are the differences between the minimal classical and quantum entropy for each dimension witness. The units of the longitudinal coordinates in all figures are \textit{bit}.\label{g-2}}
\end{figure*}

\textit{Scenario}{\textemdash}The prepare-and-measure scenario we consider is illustrated in FIG.\ref{g-1}. The state preparator with $n$ buttons is shown by the left box. When button $x\in\{1,\ldots,n\}$ is pressed, it emits a message $M$ in the classical case or a state $\rho_x$ in the quantum case. The right box is the measurement device with $l$ buttons. When button $y\in\{1,\ldots,l\}$ is pressed, it performs a measurement $M_y$ on the input state, delivering the outcome $b\in\{-1,+1\}$. $P(b|x,y)$ represents the probability for yielding the result $b$ when the measurement $M_y$ is taken on the state $\rho_x$. The expectation value of the measurement result is $E_{xy}=P(+1|x,y)-P(-1|x,y)$.

The button $x$ and $y$ are pressed upon the observers' request while the probability distributions of $P(x)$ and $P(y)$ are uniform and independent, i.e., $P(x)=1/n$ and $P(y)=1/l$. In the case of a $d$-dimensional classical system, it obeys deterministic strategies labeled by $\lambda$ in the spirit of the ontological model\cite{onto}. Hence, $E_{xy}=\sum_m\sum_{\lambda}E(y,m,\lambda)P(m|x,\lambda)q_{\lambda}$, where $q_{\lambda}$ is the probability of the strategy $\lambda$, $\sum_{\lambda}q_{\lambda}=1$, $P(m|x,\lambda)\in\{0,1\}$ and $E(y,m,\lambda)\in\{-1,+1\}$. The probability distribution of the message is $p_m=\sum_x\sum_\lambda{P(m|x,\lambda)}q_{\lambda}/n$, where $m\in\{0,\ldots,d-1\}$. The Shannon entropy of the average message $M$ is $H(M)=-\sum_{m=0}^{d-1}p_m\log_2p_m$. In the case of a $d$-dimensional quantum system, $E_{xy}=\tr(\rho_xM_y)$ where the state $\rho_x$ and the measurement $M_y$ act on $\mathbb{C}^d$. The von Neumann entropy of the average emitted state is $S(\rho)=-\tr(\rho\log_2\rho)$, where $\rho=\sum_x\rho_x/n$.



\textit{Theoretical analysis}{\textemdash}To investigate the gap between the minimal classical and quantum entropy, we propose and prove the following theorem to obtain the minimal quantum entropy under given values of a linear dimension witness $w_d$.
\begin{align}
w_d=\sum_{x=1}^n\sum_{y=1}^l\alpha_{xy}\tr(\rho_xM_y)
\end{align}
Specifically, there are three widely used linear dimension witnesses $I_3$, $I_4$, and $R_4$\cite{gallego2010prl,pawlowski2011pra},
\begin{align}
I_3=&E_{11}+E_{12}+E_{21}-E_{22}-E_{31}\\
I_4=&E_{11}+E_{12}+E_{13}+E_{21}+E_{22}-E_{23}+E_{31}-E_{32}-E_{41}\\
R_4=&E_{11}+E_{12}+E_{21}-E_{22}-E_{31}+E_{32}-E_{41}-E_{42}
\end{align}

\begin{theorem}
Given the value of a linear dimension witness $\sum_{x=1}^n\sum_{y=1}^l\alpha_{xy}\tr(\rho_xM_y)=w_d$, the minimum value of $S(\rho)$, where $\rho=({\rho_1+\ldots+\rho_n})/n$, can be obtained when $\rho_k(1{\leq}k{\leq}n)$ are all rank-1 and in $\mathbb{C}^n$.
\end{theorem}
\begin{proof}
See Sec.A of Supplementary Material.
\end{proof}

According to the theorem, we only need to consider the rank-1 states $\rho_x$ in a $n$-dimensional Hilbert space, which can be expressed as $\rho_x=|\psi_x{\rangle}{\langle}\psi_x|$, where
\begin{center}
\begin{align}
|\psi_1{\rangle}&=(1,0,\ldots,0) \nonumber\\
|\psi_2{\rangle}&=(\cos{\theta_{1,1}},e^{i{\varphi}_{1,1}}\sin{\theta_{1,1}},0,\ldots) \nonumber\\
|\psi_3{\rangle}&=(\cos{\theta_{2,1}},e^{i{\varphi}_{2,1}}\sin{\theta_{2,1}}\cos{\theta_{2,2}},e^{i{\varphi}_{2,2}}\sin{\theta_{2,1}}\sin{\theta_{2,2}},\ldots)\nonumber\\
&\ldots \nonumber\\
|\psi_n{\rangle}&=(\cos{\theta_{n-1,1}},\ldots,e^{i{\varphi}_{n-1,n-1}}\prod_{k=1}^{n-1}\sin{\theta_{n-1,k}})\label{eq-0}
\end{align}
\end{center}
Since the eigenvalue of the measurement $M_y$ is +1 or -1, the dimension witness has an upper bound of
\begin{align}
w_d=\sum_{x=1}^n\sum_{y=1}^l\alpha_{xy}\tr(\rho_xM_y){\leq}\sum_{y=1}^l\sum_k|\lambda_{yk}| \label{eq-1}
\end{align}
where $\{\lambda_{yk}\}$ are the eigenvalues of $\rho^{(y)}$ and $\rho^{(y)}=\sum_{x=1}^n\alpha_{xy}\rho_x$. The minimal quantum entropy under the given value of $\sum_{y=1}^l\sum_k|\lambda_{yk}|$ are obtained for the cases of $I_3$, $I_4$ and $R_4$ numerically using \textit{fmincon} in \textit{MATLAB}. The calculation results show that the minimal quantum entropy is a monotone increasing function of $\sum_{y=1}^l\sum_k|\lambda_{yk}|$. Due to Eq.(\ref{eq-1}), this function also expresses the relation between the minimal quantum entropy and the given value of the dimension witness. It is indicated by the blue curves in FIG.\ref{g-2}(a)$\sim$(c). On the other hand, the minimal classical entropy under given values of the dimension witness $I_3$, $I_4$ and $R_4$ are shown explicitly in Ref.\cite{chaves2015prl}. They are calculated and indicated by the red curves in FIG.\ref{g-2}(a)$\sim$(c), respectively.


\begin{figure*}[!htb]
\includegraphics[width=7.1 in]
{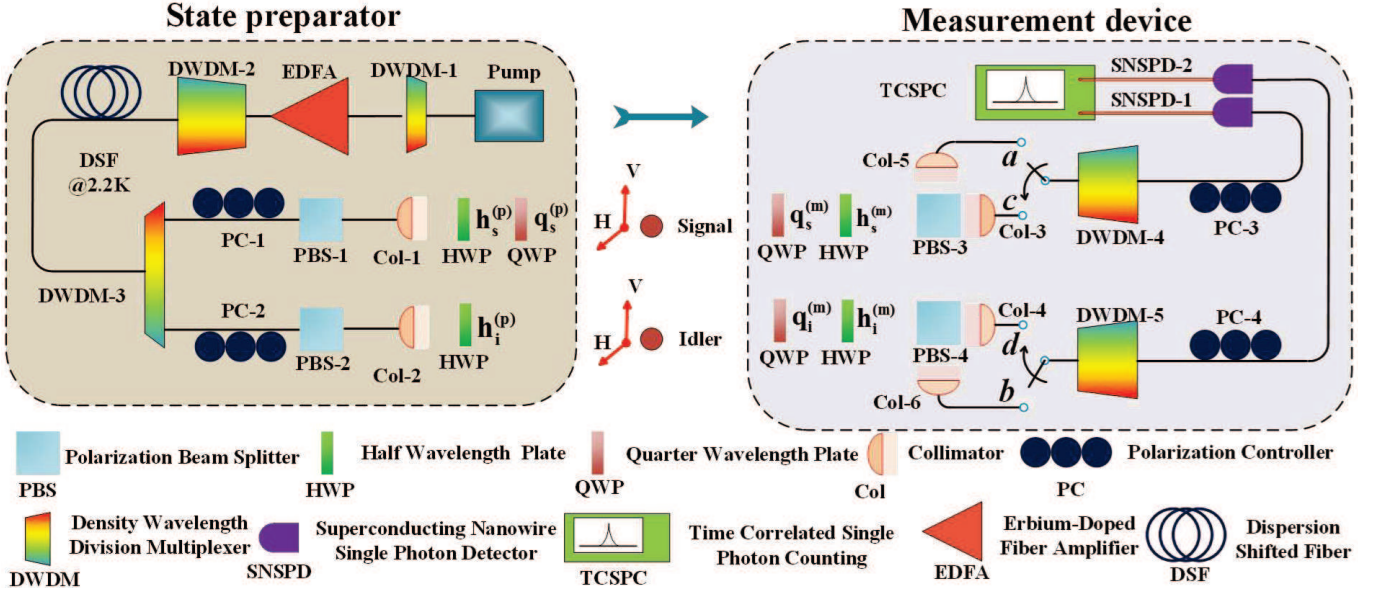}
\caption{ The experimental setup. The left part is the state preparator. The linearly polarized pulsed pump light is generated by a passive mode-locked fiber laser with a repetitive rate of 40 MHz. Its line width is narrowed to 132GHz by an optical filter (DWDM-1) with a central wavelength of 1552.52 nm. Then it is amplified by an erbium doped fiber amplifier (EDFA). The noise produced by the EDFA is suppressed by another optical filter (DWDM-2). The correlated photon pairs are generated in a piece of dispersion shifted fiber (DSF) with a length of 250 meters. It is placed in a cryostat with the superconducting nanowire single photon detectors (SNSPDs) used in this experiment and is cooled to 2.2 K to suppress the noise photons generated by the spontaneous Raman scattering. The signal and idler photons are selected and routed to two paths by the third optical filter (DWDM-3). Both of them have a linewidth of 63GHz. Two polarization controllers (PC-1 and PC-2) and two polarization beam splitters (PBS-1 and PBS-2) are used to collimate the polarization of the signal and idler photons to the vertical direction. Then, the photons are coupled to free-space by two collimators (Col-1 and Col-2). The quarter wave plate (QWP-$q_s^{(p)}$) and half wave plates (HWP-$h_s^{(p)}$ and HWP-$h_i^{(p)}$) are used to encode the information on the state of the photon pairs. The right part is the measurement device. The input photons pass through two half wave plates (HWP-$h_s^{(m)}$ and HWP-$h_i^{(m)}$) and two quarter wave plates (QWP-$q_s^{(m)}$ and QWP-$q_i^{(m)}$), then they are directed to four ports (a, b, c and d) by two polarization beam splitters (PBS-3 and PBS-4). These components are used to realize the projection measurement of the biphoton states. Four collimators (Col-3$\sim$Col-6) are used to couple the photons back to the fiber from different ports. The signal and idler photons from two specific ports are selected to be detected by two SNSPDs (fabricated by SIMIT, China). Their efficiencies and dark counts are about 40$\%$ and 80 Hz, respectively. Before the single photon detection, two additional optical filters (DWDM-4 and DWDM-5) are used to filter out the noise and two polarization controllers (PC-3 and PC-4) are used to collimate the polarizations of the photons since the efficiencies of the SNSPDs are polarization dependent. The detection events of the SNSPDs are recorded by a time correlated single photon counting module (TCSPC, PicoQuant, PicoHarp 400).}\label{g-3}
\end{figure*}

\begin{table}[htbp]
\centering
\caption{Maximum differences between minimal quantum and classical entropy for $I_3$, $I_4$ and $R_4$.}\label{t-1}
  \begin{tabular}{ cccc}
    \hline \hline
          & $H(M)$(bit) & S($\rho$)(bit) & $H(M)-S(\rho)$(bit) \\
    \hline    
    $I_3=3.622$ & 1.334 & 0.897 & 0.437  \\
    $I_4=5.760$ & 1.223 & 0.829 & 0.394  \\
    $R_4=5.211$ & 1.356 & 0.888 & 0.468  \\
    \hline
    \hline
  \end{tabular}
\end{table}

The differences between the minimal quantum and classical entropy are indicated by the green curves in FIG.\ref{g-2}(d)$\sim$(f), which show that the minimal quantum entropy is lower than the minimal classical entropy under the given value of the dimension witness. The maximum differences are presented in TABLE.\ref{t-1}. The details about the states $\rho_x$, the measurements $M_y$, the deterministic expectation values $E_{m,y}^{(\lambda)}$, the deterministic probability distribution $P_{m,x}^{(\lambda)}$ and the probability of strategy $q_\lambda$ to realize the maximum differences for the dimension witnesses are shown in Sec.B of Supplementary Material.



\textit{Experimental demonstration}{\textemdash}We encode the information on polarizations of photon pairs\cite{bogdanov2004prl,james2001pra,thew2002pra,bogdanov2004pra} generated by the spontaneous four-wave-mixing in a piece of optical fiber\cite{NTT,kumar,fiber}, by which the 3-dimensional system for the test of $I_3$ and the 4-dimensional system for the tests of $I_4$ and $R_4$ are realized. The setup is shown in FIG.\ref{g-3}.

The state preparator in FIG.\ref{g-3} emits the photon pairs with information encoded on their polarizations. The four basis states (denoted by $\{|0{\rangle},|1{\rangle},|2{\rangle},|3{\rangle}\}$)  are $\{|V{\rangle}_s|V{\rangle}_i,|H{\rangle}_s|V{\rangle}_i,|H{\rangle}_s|H{\rangle}_i,|V{\rangle}_s|H{\rangle}_i\}$, where $s$ and $i$ stand for the signal photon and the idler photon, $H$ and $V$ stand for the horizontal and vertical polarization direction. Each state is prepared by rotating the angles of the quarter wave plate and two half wave plates in the preparator, which are denoted by $q_s^{(p)}$, $h_s^{(p)}$, and $h_i^{(p)}$. The state of the photon pair can be expressed as
\begin{center}
\begin{align}
|\psi{\rangle}&=\frac{1}{\sqrt{2}}[\cos{(2q_s^{(p)}-2h_s^{(p)})}-i\cos{2h_s^{(p)}}]\cos{2h_i^{(p)}}|0{\rangle}\nonumber\\
&+\frac{1}{\sqrt{2}}[\sin{(2q_s^{(p)}-2h_s^{(p)})}-i\sin{2h_s^{(p)}}]\cos{2h_i^{(p)}}|1{\rangle}\nonumber\\
&+\frac{1}{\sqrt{2}}[\sin{(2q_s^{(p)}-2h_s^{(p)})}-i\sin{2h_s^{(p)}}]\sin{2h_i^{(p)}}|2{\rangle}\nonumber\\
&+\frac{1}{\sqrt{2}}[\cos{(2q_s^{(p)}-2h_s^{(p)})}-i\cos{2h_s^{(p)}}]\sin{2h_i^{(p)}}|3{\rangle}\label{state-p}
\end{align}
\end{center}
It is used as a 4-dimensional system for cases of $I_4$ and $R_4$. For the case of $I_3$, only the first three terms are used.

For the classical case, each state is prepared to be one of the basis states $\{|0{\rangle}$, $|1{\rangle}$, $|2{\rangle}$, or $|3{\rangle}\}$, which is perfectly distinguishable. For different strategies, different $q_\lambda$ are realized by different measurement time durations of corresponding states. The rotation angles of $q_s^{(p)}$, $h_s^{(p)}$, and $h_i^{(p)}$ for cases of $I_3$, $I_4$, and $R_4$ are shown in Sec.C of Supplementary Material.

\begin{figure*}
\begin{minipage}[b]{.235\linewidth}
  \centering
  \subfigure[Real part of $\rho$ for $I_3$]{
    \includegraphics[width=\linewidth]{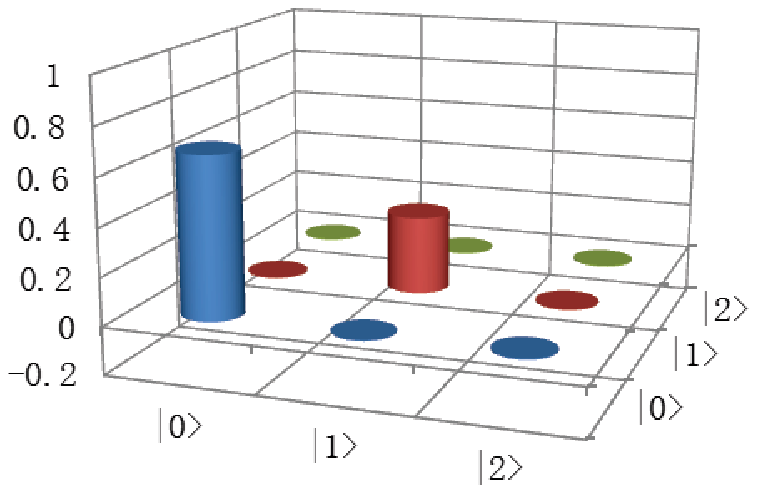}}
\end{minipage}
\begin{minipage}[b]{.235\linewidth}
  \subfigure[Imaginary part of $\rho$ for $I_3$]{
    \includegraphics[width=\linewidth]{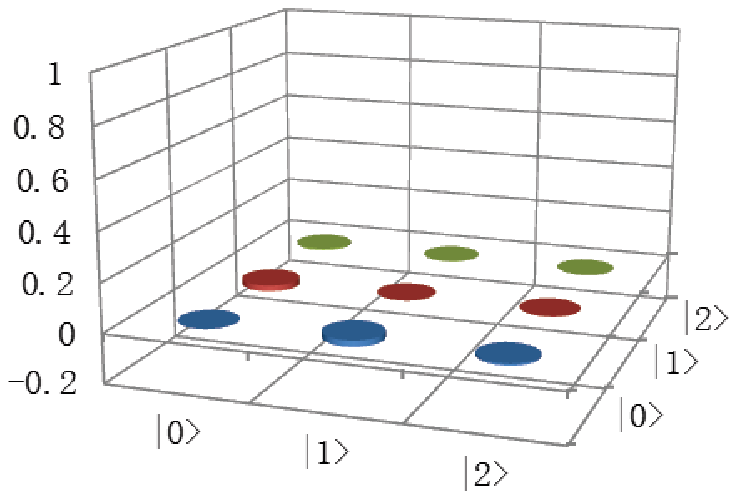}}
\end{minipage}
\begin{minipage}[b]{.21\linewidth}
  \subfigure[Distribution of $M$ for $I_3$]{
    \includegraphics[width=\linewidth]{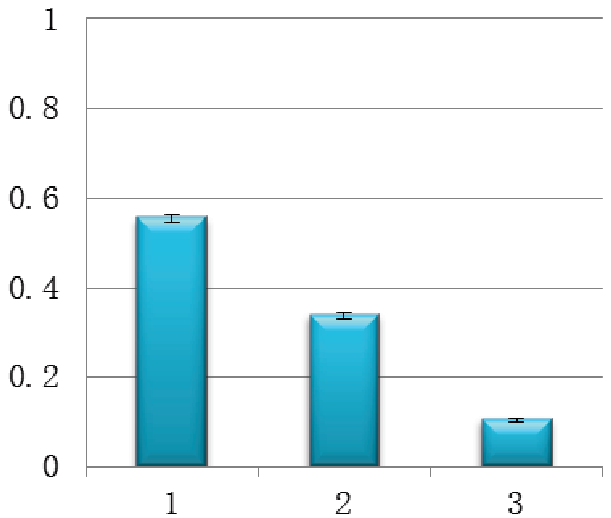}}
\end{minipage}
\begin{minipage}[b]{.3\linewidth}
  \subfigure[Results for $I_3$]{  \begin{tabular}[b]{ ccc}
    \hline \hline
          &  Value$_{th}$ & Value$_{exp}$  \\
    \hline
$w_d^{(c)} $ &   3.62  & 3.56(6)\\
$H(M)\text{ (bit) }$ &    1.33&1.34(2)\\
$w_d^{(q)}$   &3.62 &3.56(11)\\
$S(\rho)\text{ (bit) }$ &0.90&0.94(3)\\
   \hline
    \hline
  \end{tabular}}
  \label{fig:subfig} 
\end{minipage}

\begin{minipage}[b]{.235\linewidth}
  \centering
  \subfigure[Real part of $\rho$ for $I_4$]{
    \includegraphics[width=\linewidth]{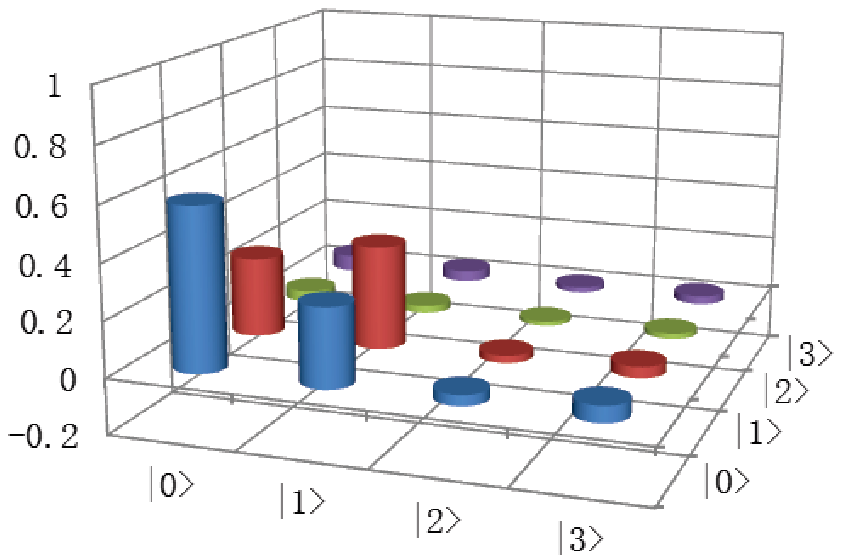}}
\end{minipage}
\begin{minipage}[b]{.235\linewidth}
  \subfigure[Imaginary part of $\rho$ for $I_4$]{
    \includegraphics[width=\linewidth]{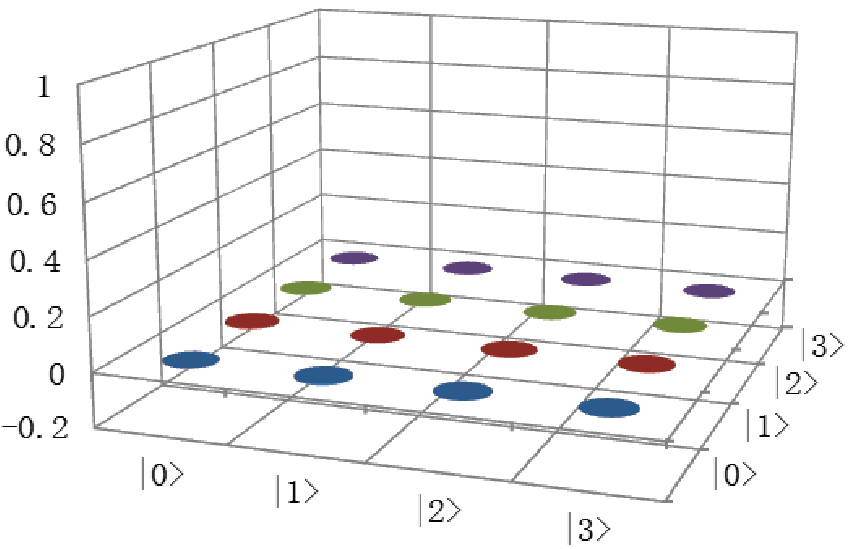}}
\end{minipage}
\begin{minipage}[b]{.21\linewidth}
  \subfigure[Distribution of $M$ for $I_4$]{
    \includegraphics[width=\linewidth]{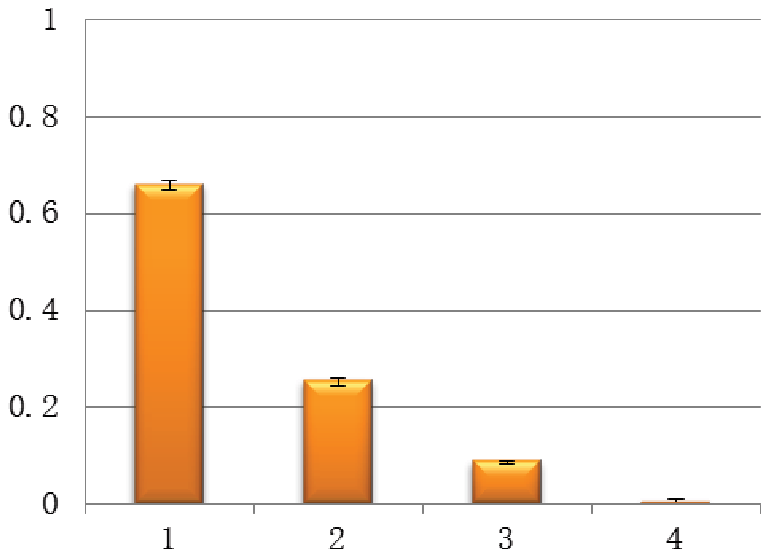}}
\end{minipage}
\begin{minipage}[b]{.3\linewidth}
  \subfigure[Results for $I_4$]{  \begin{tabular}[b]{ ccc}
    \hline \hline
          &  Value$_{th}$ & Value$_{exp}$  \\
    \hline
$w_d^{(c)} $ &      5.76 &5.67(7)\\
$H(M)\text{ (bit) }$ &1.22&1.22(6)\\
$w_d^{(q)}$   &5.76&5.63(13)\\
$S(\rho)\text{ (bit) }$ &0.83&0.88(7)\\
   \hline
    \hline
  \end{tabular}}
  \label{fig:subfig} 
\end{minipage}

\begin{minipage}[b]{.235\linewidth}
  \centering
  \subfigure [Real part of $\rho$ for $R_4$]{
    \includegraphics[width=\linewidth]{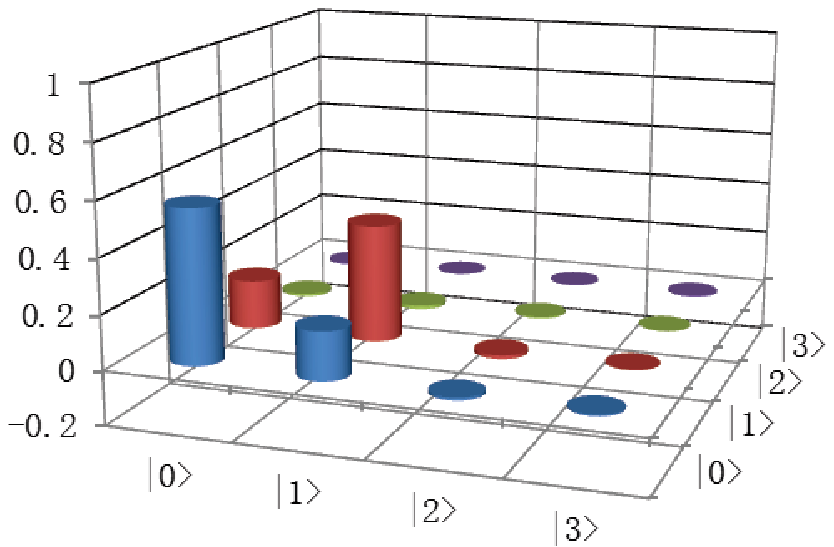}}
\end{minipage}
\begin{minipage}[b]{.235\linewidth}
  \subfigure[Imaginary part of $\rho$ for $R_4$]{
    \includegraphics[width=\linewidth]{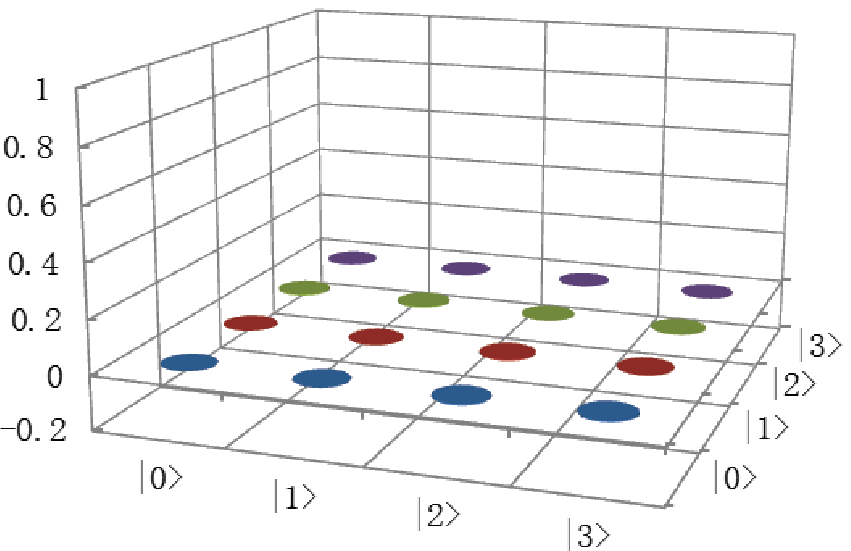}}
\end{minipage}
\begin{minipage}[b]{.21\linewidth}
  \subfigure[Distribution of $M$ for $R_4$]{
    \includegraphics[width=\linewidth]{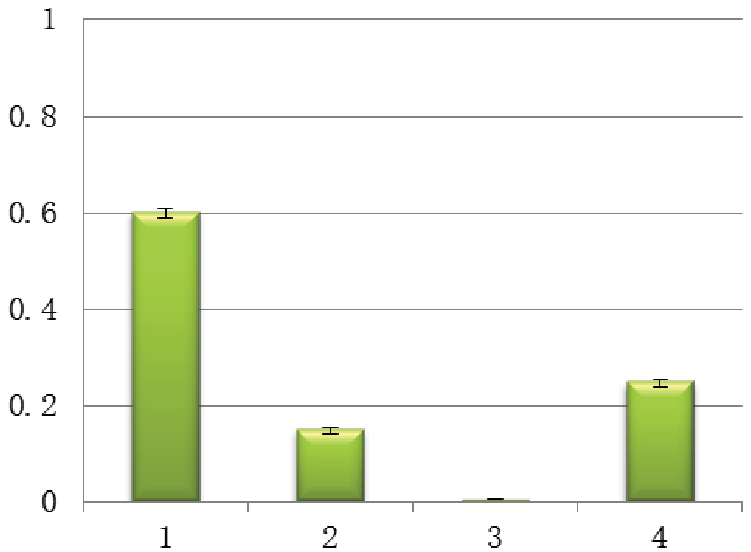}}
\end{minipage}
\begin{minipage}[b]{.3\linewidth}
  \subfigure[Results for $R_4$]{\begin{tabular}[b]{ ccc}
    \hline \hline
          &  Value$_{th}$ & Value$_{exp}$  \\
    \hline
$w_d^{(c)} $ &  5.21 &   5.08(10)\\
$H(M)\text{ (bit) }$ &1.36 & 1.38(4)\\
$w_d^{(q)}$   & 5.21 &5.16(24)\\
$S(\rho)\text{ (bit) }$ &0.89 &0.96(7)\\
   \hline
    \hline
  \end{tabular}}
  \label{fig:subfig} 
\end{minipage}

\caption{The experimental results of the state $\rho$ and the message $M$ for $I_3$, $I_4$ and $R_4$.}\label{g-4}
\end{figure*}

The right part of FIG.\ref{g-3} is the measurement device, which realizes the projection measurements of the state, by which the dimension witness and entropy for quantum and classical cases can be measured. The coincidence count of the two detectors is denoted by $D_{a,b}$ if the photons from port $a$ and $b$ are detected simultaneously. Similarly, $D_{a,d}$, $D_{c,b}$, and $D_{c,d}$ are the coincidence counts of the photons from the corresponding ports. For the quantum dimension witness, $P(-1|x,y)$ is obtained by $D_{a,b}$, and $P(+1|x,y)$ is obtained by $D_{a,d}$, $D_{c,b}$, and $D_{c,d}$. The projection state $|m_y{\rangle}$ is produced by rotating angles of two quarter wave plates and two half wave plates in the measurement device, which are denoted by $q_s^{(m)}$, $q_i^{(m)}$, $h_s^{(m)}$, and $h_i^{(m)}$, respectively. $|m_y{\rangle}$ can be expressed as
\begin{widetext}
\begin{align}
|m_y{\rangle}&=\frac{1}{2}[\cos{(2q_s^{(m)}-2h_s^{(m)})}+i\cos{2h_s^{(m)}}][\cos{(2q_i^{(m)}-2h_i^{(m)})}+i\cos{2h_i^{(m)}}]|0{\rangle}\nonumber\\
&+\frac{1}{2}[\sin{(2q_s^{(m)}-2h_s^{(m)})}+i\sin{2h_s^{(m)}}][\cos{(2q_i^{(m)}-2h_i^{(m)})}+i\cos{2h_i^{(m)}}]|1{\rangle}\nonumber\\
&+\frac{1}{2}[\sin{(2q_s^{(m)}-2h_s^{(m)})}+i\sin{2h_s^{(m)}}][\sin{(2q_i^{(m)}-2h_i^{(m)})}+i\sin{2h_i^{(m)}}]|2{\rangle}\nonumber\\
&+\frac{1}{2}[\cos{(2q_s^{(m)}-2h_s^{(m)})}+i\cos{2h_s^{(m)}}][\sin{(2q_i^{(m)}-2h_i^{(m)})}+i\sin{2h_i^{(m)}}]|3{\rangle}_i\label{state-m}
\end{align}
\end{widetext}
For the case of $I_3$, only the first three terms are used.

For the measurement of quantum entropy witness, the states are reconstructed by the quantum state tomography\cite{james2001pra,thew2002pra,bogdanov2004pra} which is realized by detect-events of $D_ {a,b}$ under different projection states. The details about the rotation angles of $h_s^{(m)}$, $q_s^{(m)}$, $h_i^{(m)}$, and $q_i^{(m)}$ for the quantum dimension witness and entropy are shown in Sec.C of Supplementary Materials. For the classical dimension witness and entropy, the angles of $h_s^{(m)}$, $q_s^{(m)}$, $h_i^{(m)}$, and $q_i^{(m)}$ are all set to $0^{\circ}$. The measurement settings are reduced to an arrangement that each coincidence count indicates a specific basis state, i.e., $|0{\rangle}{\rightarrow}D_{a,b}$, $|1{\rangle}{\rightarrow}D_{c,b}$, $|2{\rangle}{\rightarrow}D_{c,d}$, and $|3{\rangle}{\rightarrow}D_{a,d}$.

In the experiment, both the signal and idler photon count rates are about 19 kHz. The coincidence count rate is about 900 $s^{-1}$. The generation rate of the photon pairs is a little less than 0.01/pulse and the coincidence and accidence ratio (CAR) is higher than 100. Both the collection efficiencies of the signal and idler photons are about 5$\%$, including the optical losses and the detector efficiencies. The time window of the coincidence counting is 300 ps. For the quantum case, each counting time of an event $E_{xy}$ is 30s. For the classical case, the total counting time of each event $E_{xy}$ is 30s, and the counting time for each strategy $\lambda$ is $q_\lambda\times$30s.

The experimental results are shown in FIG.\ref{g-4}. FIG.\ref{g-4}(a) and FIG.\ref{g-4}(b) are the real and imaginary parts of the density matrix $\rho$ for the case of $I_3$, which is reconstructed by the measurement of quantum state tomography. FIG.\ref{g-4}(c) is the measured distribution of message $M$ for the case of $I_3$. The quantum entropy $S(\rho)$ and the classical entropy $H(M)$ are calculated according to FIG.\ref{g-4}(a), (b) and (c) and shown in FIG.\ref{g-4}(d), with the experimental results of quantum and classical dimension witness ($w_d^{(q)}$ and $w_d^{(c)}$) for the case of $I_3$. The theoretical values of $S(\rho)$, $H(M)$, $w_d^{(q)}$ and $w_d^{(c)}$ are also listed in FIG.\ref{g-4}(d) for comparison. For the cases of $I_4$ and $R_4$, the corresponding results are shown in FIG.\ref{g-4}(e)$\sim$(h) and FIG.\ref{g-4}(i)$\sim$(l), respectively. The unideal factors in the experiment are analyzed. The errors of the experiment results are calculated and shown in FIG.\ref{g-4}(d), (h) and (l), considering the error sources of the limited angle precision of the polarization components, the imperfection of the polarization splitting and the propagated Poissonian counting statistics of the detection events. It can be seen that the experimental results agree well with the theoretical expectations, showing that the minimal quantum entropy are lower than the minimal classical entropy under given values of the dimension witness in all the cases.



{\it Discussion}{\textemdash}In the theoretical analysis we have proved that the use of a system with the dimension higher than $n$ is not helpful to reduce the minimal quantum entropy under given values of the linear dimension witness $\sum_{x=1}^n\sum_{y=1}^l\alpha_{xy}\tr(\rho_xM_y)$. An related question is that if the given value of the dimension witness can be obtained by a $d$-dimensional system, where $d<n$, whether the minimal value of $S(\rho)$ could also be obtained by the $d$-dimensional system? On the other hand, we have calculated the minimal classical entropy according to Eq.(11) in Ref.\cite{chaves2015prl} for the dimension witness of $I_3$, $I_4$, and $R_4$ in the theoretical part of this paper. However, for arbitrary linear dimension witness, could the minimal classical entropy be obtained by the same way? We find that the answers of two above questions are ``no''. We list counter-examples for them in Sec.D of Supplementary Materials. It can be expected that $\alpha_{xy}$ would determine whether they hold or not, however, the condition of $\alpha_{xy}$ to support them are not clear. It is an interesting open problem.



{\it Conclusion}{\textemdash}We propose and prove a theorem which claims that the minimal value of $S(\rho)$ under given values of the linear dimension witness $\sum_{x=1}^n\sum_{y=1}^l\alpha_{xy}\tr(\rho_xM_y)$ can be obtained in $\mathbb{C}^n$. This theorem is used to obtain the minimal quantum entropy for $I_3$, $I_4$ and $R_4$. With the minimal classical entropy indicated in Ref.\cite{chaves2015prl}, the differences between the minimal quantum and classical entropy are illustrated. Then we experimentally verify it by a telecom band biphoton system, in which the photon pair generation is based on the spontaneous four-wave-mixing in optical fibers and the single photon detections are based on SNSPDs. The qutrit and ququart are encoded on the polarizations of the photon pairs. The experimental results agree well with the theoretical values, demonstrating the reduction of communication entropy from classical to quantum system.


\begin{acknowledgments}
This work was supported by 973 Programs of China under Contract No. 2013CB328700 and 2011CBA00303, the National Natural Science Foundation of China under Contract No. 61575102, 91121022 and 61321004, Tsinghua University Initiative Scientific Research Program under Contract No. 20131089382. Strategic Priority Research Program (B) of the Chinese Academy of Sciences (XDB04020100).
\end{acknowledgments}


\newpage

\begin{center}
\bf{Supplementary Material:}
\end{center}


\subsection{A. \bf{Proof of the Theorem}}

\begin{lemma}

Let $M$ and $\rho$ be an observable and a density matrix, respectively, where $\rank(\rho)=2$. Then there exist two density matrices $\rho_0$, $\rho_1$ and two positive real numbers $\mu_0$, $\mu_1$, subject to
\begin{align}
&\mu_0+\mu_1=1 \label{lemma_1_1}\\
&\mu_0\rho_0+\mu_1\rho_1=\rho \label{lemma_1_2}\\
&\rank(\rho_0)=\rank(\rho_1)=1 \label{lemma_1_3}\\
&\tr(\rho_0M)=\tr(\rho_1M) \label{lemma_1_4}
\end{align}

\end{lemma}


\begin{proof}

Since the density matrix $\rho$ is an Hermitian matrix, it can be represented by a diagonal matrix $\Lambda$ under a specific complete orthogonal basis. Let the complete orthogonal basis and the diagonal matrix $\Lambda$ be $\{|0{\rangle},|1{\rangle},|2{\rangle},\dots\}$ and $\diag\{\lambda_0,\lambda_1,\lambda_2,\ldots\}$, respectively. Since $\rank(\rho)=2$, without loss of generality, let $\lambda_{k}=0 \text{ while } k{\geq}2$. Then $\lambda_{0}>0$ and $\lambda_1>0$. Hence, the density matrix $\rho$ can be written as $\rho=\lambda_0|0{\rangle}{\langle}0|+\lambda_1|1{\rangle}{\langle}1|$. The observable $M$ can be written as $M={\sum_{k=0}}{\sum_{t=0}}{m_{kt}|k{\rangle}{\langle}t|}$. Without loss of generality, let $m_{00}{\geq}m_{11}$.\\


{\noindent}Case 1. \ \  $m_{00}=m_{11}$\\

Let
\begin{align}
&\rho_0=|0{\rangle}{\langle}0|\\
&\rho_1=|1{\rangle}{\langle}1|\\
&\mu_0=\lambda_0\\
&\mu_1=\lambda_1
\end{align}

Then Eq.(\ref{lemma_1_1}) holds since the trace of the density matrix $\rho$ is 1. Eq.(\ref{lemma_1_2})$\thicksim$(\ref{lemma_1_3}) hold clearly. Eq.(\ref{lemma_1_4}) holds since $\tr(\rho_0M)=m_{00}=m_{11}=\tr(\rho_1M)$.\\


{\noindent}Case 2. \ \  $m_{00}>m_{11}$ \\

Let's define a function
\begin{align}
f(\theta)=m_{00}\cos^2{\theta}+m_{11}\sin^2{\theta}+(m_{10}+m_{01})\cos{\theta}\sin{\theta}-m_{00}\lambda_0-m_{11}\lambda_1
\end{align}
Since $m_{00}>m_{11}$ and $\lambda_k>0\text{ while }k{\in}\{0,1\}$, $f(0)=\lambda_1(m_{00}-m_{11})>0$ and $f(\frac{\pi}{2})=f(-\frac{\pi}{2})=-\lambda_0(m_{00}-m_{11})<0$. Since $f(\theta)$ is a continuous function, by the intermediate value theorem there exist $\theta_1\in(0,\frac{\pi}{2})$ and $\theta_2\in(-\frac{\pi}{2},0)$ such that $f(\theta_1)=f(\theta_2)=0$.

Then let

\begin{align}
&\rho_0=(\cos\theta_1|0{\rangle}+\sin\theta_1|1{\rangle})(\cos\theta_1{\langle}0|+\sin\theta_1{\langle}1|)\\
&\rho_1=(\cos\theta_2|0{\rangle}+\sin\theta_2|1{\rangle})(\cos\theta_2{\langle}0|+\sin\theta_2{\langle}1|)\\
&\mu_0=\frac{-\sin{2\theta_2}}{{\sin{2\theta_1}-\sin{2\theta_2}}}\\
&\mu_1=\frac{\sin{2\theta_1}}{{\sin{2\theta_1}-\sin{2\theta_2}}}
\end{align}

Since $2\theta_1\in(0,\pi)$ and $2\theta_1\in(-\pi,0)$, $\sin{2\theta_1}>0$, and $\sin2\theta_2<0$. It follows that $\mu_0>0$ and $\mu_1>0$. Furthermore, both Eq.(\ref{lemma_1_1}) and Eq.(\ref{lemma_1_3}) hold clearly. Eq.(\ref{lemma_1_4}) also holds since $\tr(\rho_0M)=f(\theta_1)+\lambda_0m_{00}+\lambda_1m_{11}=\lambda_0m_{00}+\lambda_1m_{11}=f(\theta_2)+\lambda_0m_{00}+\lambda_1m_{11}=\tr(\rho_1M)$.

Consider that
\begin{align}
0&=\mu_0\cdot0+\mu_1\cdot0\nonumber\\
&=\mu_0f(\theta_1)+\mu_1f(\theta_2)\nonumber\\
&=(\mu_0\cos^2\theta_1+\mu_1\cos^2\theta_2-\lambda_0)m_{00}+(\mu_0\sin^2\theta_1+\mu_1\sin^2\theta_2-\lambda_1)m_{11}\nonumber\\
&=(\mu_0\cos^2\theta_1+\mu_1\cos^2\theta_2-\lambda_0)m_{00}+[(1-\mu_0\cos^2\theta_1-\mu_1\cos^2\theta_2)-(1-\lambda_0)]m_{11}\nonumber\\
&=(\mu_0\cos^2\theta_1+\mu_1\cos^2\theta_2-\lambda_0)(m_{00}-m_{11})\nonumber\\
&=(\frac{-\sin{2\theta_2}\cos^2{\theta_1}+\sin{2\theta_1}\cos^2{\theta_2}}{{\sin{2\theta_1}-\sin{2\theta_2}}}-\lambda_0)(m_{00}-m_{11})\label{lemma_1_asistant_2}
\end{align}

Then
\begin{align}
\frac{-\sin{2\theta_2}\cos^2{\theta_1}+\sin{2\theta_1}\cos^2{\theta_2}}{{\sin{2\theta_1}-\sin{2\theta_2}}}=\lambda_0\label{lemma_1_asistant_11}\\  \frac{-\sin{2\theta_2}\sin^2{\theta_1}+\sin{2\theta_1}\sin^2{\theta_2}}{{\sin{2\theta_1}-\sin{2\theta_2}}}=\lambda_1\label{lemma_1_asistant_12}
\end{align}

Hence
\begin{align}
\mu_0\rho_0+\mu_1\rho_1&=\mu_0(\cos\theta_1|0{\rangle}+\sin\theta_1|1{\rangle})(\cos\theta_1{\langle}0|+\sin\theta_1{\langle}1|)+\mu_1(\cos\theta_2|0{\rangle}+\sin\theta_2|1{\rangle})(\cos\theta_2{\langle}0|+\sin\theta_2{\langle}1|)\nonumber\\ &=\frac{-\sin{2\theta_2}\cos^2{\theta_1}+\sin{2\theta_1}\cos^2{\theta_2}}{{\sin{2\theta_1}-\sin{2\theta_2}}}|0{\rangle}{\langle}0|+\frac{-\sin{2\theta_2}\sin^2{\theta_1}+\sin{2\theta_1}\sin^2{\theta_2}}{{\sin{2\theta_1}-\sin{2\theta_2}}}|1{\rangle}{\langle}1| \label{lemma_1_asistant_1}
\end{align}

Since Eq.(\ref{lemma_1_asistant_11})$\sim$(\ref{lemma_1_asistant_1}), Eq.(\ref{lemma_1_2}) holds.
\end{proof}


\begin{lemma} Let $M$ and $\rho$ be an observable and a density matrix, respectively, where $\rank(\rho)=n>2$. Then there exist three density matrices $\rho_0$, $\rho_1$, $\rho'$ and three positive real numbers $\mu_0$, $\mu_1$ and $\mu'$, subject to

\begin{align}
&\mu_0+\mu_1+\mu'=1 \label{lemma_2_1}\\
&\mu_0\rho_0+\mu_1\rho_1+\mu'\rho'=\rho \label{lemma_2_2}\\
&\rank(\rho_0)=\rank(\rho_1)=1 \label{lemma_2_3}\\
&\tr(\rho_0M)=\tr(\rho_1M)=\tr(\rho{M}) \label{lemma_2_5}\\
&\rank(\rho')<\rank(\rho) \label{lemma_2_4}
\end{align}

\end{lemma}


\begin{proof} Since the density matrix $\rho$ is an Hermitian matrix, it can be represented by a diagonal matrix $\Lambda$ under a specific complete orthogonal basis. Let the complete orthogonal basis be $\{|0{\rangle},|1{\rangle},|2{\rangle},\dots,|n-2{\rangle},|n-1{\rangle},|n{\rangle},\ldots\}$ and the diagonal matrix $\Lambda$ be
$\diag\{\lambda_0,\lambda_1,\lambda_2,\ldots,\lambda_{n-2},\lambda_{n-1},\lambda_{n},\ldots\}$. Since $\rank(\rho)=n$, without loss of generality, let $\lambda_{k}=0 \text{ while } k{\geq}n$. Then $\lambda_{k}>0 \text{ while } 0{\leq}k{\leq}{n-1}$. Hence, the density matrix $\rho$ can be written as $\rho={\sum_{k=0}^{n-1}}{\lambda_k|k{\rangle}{\langle}k|}$. The observable $M$ can be written as $M={\sum_{k=0}}{\sum_{t=0}}{m_{kt}|k{\rangle}{\langle}t|}$. Let $m_{ii}$ and $m_{jj} $ be the maximum and minimum among the first $n$ diagonal elements of the matrix of $M$, hence
\begin{eqnarray}
m_{ii}=\max\{m_{00},\ldots,{m_{n-1n-1}}\} \label{lemma_1_max}\\
m_{jj}=\min\{m_{00},\ldots,{m_{n-1n-1}}\} \label{lemma_1_min}
\end{eqnarray}
It follows that
\begin{eqnarray}
m_{ii}{\geq}m_{jj}
\end{eqnarray}


{\noindent}Case 1. \ \  $m_{ii}=m_{jj}$\\

Thus $m_{00}=m_{11}=...=m_{n-1n-1}$. Let
\begin{align}
&\rho_0=|i{\rangle}{\langle}i|\\
&\rho_1=|j{\rangle}{\langle}j|\\
&\rho'=\frac{1}{1-\lambda_i-\lambda_j}\sum_{\substack{k=0 \\ k{\neq}i \\ k{\neq}j}}^{n-1}{\lambda_k}|k{\rangle}{\langle}k|\\
&\mu_0=\lambda_i\\
&\mu_1=\lambda_j\\
&\mu'=1-\lambda_i-\lambda_j
\end{align}

Since $\rank(\rho){>}2$, $\mu_0$, $\mu_1$, and $\mu'$ are all positive real numbers. $\rho_0$, $\rho_1$, and $\rho'$ are density matrices since $\lambda_k{\geq}0$ and $\sum_{k=0,k{\neq}i,k{\neq}j}^{n-1}\lambda_k=1-\lambda_i-\lambda_j$. Eq.(\ref{lemma_2_1})$\thicksim$(\ref{lemma_2_3}) hold clearly. $\rank(\rho')=n-2<n=\rank(\rho)$, hence Eq.(\ref{lemma_2_4}) holds. $\tr(\rho_1M)=m_{ii}$, $\tr(\rho_2M)=m_{jj}=m_{ii}$ and $\tr(\rho{M})=\sum_{k=0}^{n-1}{\lambda_km_{kk}}=m_{ii}$ due to that $\sum_{k=0}^{n-1}\lambda_k=1$ and $m_{00}=m_{11}=\ldots=m_{n-2n-2}=m_{n-1n-1}$. It follows that Eq.(\ref{lemma_2_5}) holds.\\


{\noindent}Case 2. \ \  $m_{ii}>m_{jj}$ \\

Let's define a function
\begin{align}
F(\theta)=m_{ii}\cos^2{\theta}+m_{jj}\sin^2{\theta}+(m_{ji}+m_{ij})\cos{\theta}\sin{\theta}-\sum_{k=0}^{n-1}{m_{kk}\lambda_k}
\end{align}
Since $m_{ii}>m_{jj}$, $\lambda_k>0 \text{ while } 0{\leq}k{\leq}n-1$, and $\sum_{k=0}^{n-1}\lambda_k=1$, $F(0)=m_{ii}-\sum_{k=0}^{n-1}{\lambda_km_{kk}}>0$ and $F(\frac{\pi}{2})=F(-\frac{\pi}{2})=m_{jj}-\sum_{k=0}^{n-1}{\lambda_km_{kk}}<0$. Since $F(\theta)$ is a continuous function, by the intermediate value theorem there exist $\theta_1\in(0,\frac{\pi}{2})$ and $\theta_2\in(-\frac{\pi}{2},0)$ such that $F(\theta_1)=F(\theta_2)=0$.\\


{\noindent} Case 2.1
$({\sin{2\theta_1}\cos^2{\theta_2}-\sin{2\theta_2}\cos^2{\theta_1}})/{\lambda_i}>({\sin{2\theta_1}\sin^2{\theta_2}-\sin{2\theta_2}\sin^2{\theta_1}})/{\lambda_j}$
\\

Let \begin{align}
&\rho_0=(\cos\theta_1|i{\rangle}+\sin\theta_1|j{\rangle})(\cos\theta_1{\langle}i|+\sin\theta_1{\langle}j|)\\
&\rho_1=(\cos\theta_2|i{\rangle}+\sin\theta_2|j{\rangle})(\cos\theta_2{\langle}i|+\sin\theta_2{\langle}j|)\\
&\mu_0=\lambda_i\frac{-\sin{2\theta_2}}{{\sin{2\theta_1}\cos^2{\theta_2}-\sin{2\theta_2}\cos^2{\theta_1}}}\\
&\mu_1=\lambda_i\frac{\sin{2\theta_1}}{{\sin{2\theta_1}\cos^2{\theta_2}-\sin{2\theta_2}\cos^2{\theta_1}}}\\
&\mu'=1-\lambda_i\frac{\sin{2\theta_1}-\sin{2\theta_2}}{{\sin{2\theta_1}\cos^2{\theta_2}-\sin{2\theta_2}\cos^2{\theta_1}}}\\
&\rho'=\frac{1}{\mu'}[(\sum_{\substack{k=0 \\ k{\neq}i \\ k{\neq}j}}^{n-1}{\lambda_k}|k{\rangle}{\langle}k|)+(\lambda_j-\lambda_i\frac{{\sin{2\theta_1}\sin^2{\theta_2}-\sin{2\theta_2}\sin^2{\theta_1}}}{{\sin{2\theta_1}\cos^2{\theta_2}-\sin{2\theta_2}\cos^2{\theta_1}}})|j{\rangle}{\langle}j|]
\end{align}

Since $2\theta_1\in(0,\pi)$ and $2\theta_1\in(-\pi,0)$, $\sin{2\theta_1}>0$, and $\sin2\theta_2<0$. It follows that $\mu_0>0$ and $\mu_1>0$. Furthermore, $\mu'>0$ since
\begin{align}
\mu'&=1-\lambda_i\frac{\sin{2\theta_1}-\sin{2\theta_2}}{{\sin{2\theta_1}\cos^2{\theta_2}-\sin{2\theta_2}\cos^2{\theta_1}}}\nonumber\\
&=(\sum_{\substack{k=0\\k{\neq}i\\k{\neq}j}}^{n-1}{\lambda_k}+\lambda_j+\lambda_i)-\lambda_i\frac{\sin{2\theta_1}-\sin{2\theta_2}}{{\sin{2\theta_1}\cos^2{\theta_2}-\sin{2\theta_2}\cos^2{\theta_1}}}\nonumber\\
&=(\sum_{\substack{k=0\\k{\neq}i\\k{\neq}j}}^{n-1}{\lambda_k})+(\lambda_j-\lambda_i\frac{{\sin{2\theta_1}\sin^2{\theta_2}-\sin{2\theta_2}\sin^2{\theta_1}}}{{\sin{2\theta_1}\cos^2{\theta_2}-\sin{2\theta_2}\cos^2{\theta_1}}})\nonumber\\
&>\lambda_j-\lambda_i\frac{{\sin{2\theta_1}\sin^2{\theta_2}-\sin{2\theta_2}\sin^2{\theta_1}}}{{\sin{2\theta_1}\cos^2{\theta_2}-\sin{2\theta_2}\cos^2{\theta_1}}}\nonumber\\
&>0
\end{align}

Since $\rho'$ is a semi-positive definite Hermitian matrix and $\tr(\rho')={\mu'}/{\mu'}=1$, $\rho'$ is a density matrix. Eq.(\ref{lemma_2_1}) and Eq.(\ref{lemma_2_2}) hold clearly. $\rho_0$ and $\rho_1$ are rank-1 density matrices, then Eq.(\ref{lemma_2_3}) holds. Eq.(\ref{lemma_2_5}) also holds since $\tr(\rho_0M)=F(\theta_1)+\sum_{k=0}^{n-1}\lambda_km_{kk}=0+\tr({\rho}M)=F(\theta_2)+\sum_{k=0}^{n-1}\lambda_km_{kk}=\tr(\rho_1M)$. $\rank(\rho')<n$ since $\rho'$ doesn't have the term of $|i{\rangle}{\langle}i|$. Then Eq.(\ref{lemma_2_4}) holds. \\


{\noindent} Case 2.2
$({\sin{2\theta_1}\cos^2{\theta_2}-\sin{2\theta_2}\cos^2{\theta_1}})/{\lambda_i}{\leq}({\sin{2\theta_1}\sin^2{\theta_2}-\sin{2\theta_2}\sin^2{\theta_1}})/{\lambda_j}$\\

Let
\begin{align}
&\rho_0=(\cos\theta_1|i{\rangle}+\sin\theta_1|j{\rangle})(\cos\theta_1{\langle}i|+\sin\theta_1{\langle}j|)\\
&\rho_1=(\cos\theta_2|i{\rangle}+\sin\theta_2|j{\rangle})(\cos\theta_2{\langle}i|+\sin\theta_2{\langle}j|)\\
&\mu_0=\lambda_j\frac{-\sin{2\theta_2}}{{\sin{2\theta_1}\sin^2{\theta_2}-\sin{2\theta_2}\sin^2{\theta_1}}}\\
&\mu_1=\lambda_j\frac{\sin{2\theta_1}}{{\sin{2\theta_1}\sin^2{\theta_2}-\sin{2\theta_2}\sin^2{\theta_1}}}\\
&\mu'=1-\lambda_j\frac{\sin{2\theta_1}-\sin{2\theta_2}}{{\sin{2\theta_1}\sin^2{\theta_2}-\sin{2\theta_2}\sin^2{\theta_1}}}\\
&\rho'=\frac{1}{\mu'}[(\sum_{\substack{k=0 \\ k{\neq}i \\ k{\neq}j}}^{n-1}{\lambda_k}|k{\rangle}{\langle}k|)+(\lambda_i-\lambda_j\frac{{\sin{2\theta_1}\cos^2{\theta_2}-\sin{2\theta_2}\cos^2{\theta_1}}}{{\sin{2\theta_1}\sin^2{\theta_2}-\sin{2\theta_2}\sin^2{\theta_1}}})|i{\rangle}{\langle}i|]
\end{align}

Eq.(\ref{lemma_2_1})$\sim$(\ref{lemma_2_4}) hold by a proof similar to the Case 2.1.
\end{proof}


\begin{lemma} Let $M$ and $\rho$ be an observable and a density matrix, respectively. Then there exist density matrices $\{\rho_0,\ldots,\rho_{s-1}\}$ and positive real numbers $\{\nu_0,\ldots,\nu_{s-1}\}$, subject to
\begin{align}
&\sum_{k=0}^{s-1}\nu_k=1 \label{lemma_3_1}\\
&\sum_{k=0}^{s-1}\nu_k\rho_k=\rho \label{lemma_3_2}\\
&\rank(\rho_k)=1 \text{ while } 0{\leq}k{\leq}s-1 \label{lemma_3_3}\\
&\tr(\rho_k{M})=\tr(\rho{M}) \text{ while } 0{\leq}k{\leq}s-1 \label{lemma_3_4}
\end{align}

\end{lemma}


\begin{proof}
{\noindent}Case 1. \ \  $\rank(\rho)=1$\\

Let $\rho_1=\rho$, $\nu_1=1$. Eq.(\ref{lemma_3_1})$\sim$(\ref{lemma_3_4}) hold.\\


{\noindent}Case 2. \ \  $\rank(\rho)=2$\\

Using Lemma 1, there exist $\rho_0$, $\rho_1$, $\mu_0$, and $\mu_1$ satisfying Eq.(\ref{lemma_1_1})$\sim$(\ref{lemma_1_4}). Then let $\nu_0=\mu_0$ and $\nu_1=\mu_1$, Eq.(\ref{lemma_3_1})$\sim$(\ref{lemma_3_3}) hold. Consider that
\begin{align}
\tr(\rho{M})=&\tr\{(\mu_0\rho_0+\mu_1\rho_1){M}\}\nonumber\\
=&\mu_0\tr(\rho_0{M})+\mu_1\tr(\rho_1{M})\nonumber\\
=&\mu_0\tr(\rho_0{M})+(1-\mu_0)\tr(\rho_1{M})\nonumber\\
=&\tr(\rho_0{M})
\end{align}

Then Eq.(\ref{lemma_3_4}) holds. \\


{\noindent}Case 3. \ \  $\rank(\rho)>2$\\

Using Lemma 2, there exist $\rho_0$, $\rho_1$, $\rho'$, $\mu_0$, $\mu_1$, and $\mu'$ satisfying Eq.(\ref{lemma_2_1})$\sim$(\ref{lemma_2_4}). If $\rank(\rho')$ is still larger than 2, using Lemma 2 again. There exist $\rho_2$, $\rho_3$, $\rho''$, $\mu_2$, $\mu_3$, and $\mu''$, subject to

\begin{align}
&\mu_2+\mu_3+\mu''=1 \label{lemma_3_6}\\
&\rho'=\mu_2\rho_2+\mu_3\rho_3+\mu''\rho'' \label{lemma_3_7}\\
&\rank(\rho_2)=\rank(\rho_3)=1 \label{lemma_3_8}\\
&\tr(\rho_2M)=\tr(\rho_3M)=\tr(\rho'{M}) \label{lemma_3_10}\\
&\rank(\rho'')<\rank(\rho') \label{lemma_3_9}
\end{align}

Repeat using Lemma 2 until $\rank(\rho{'}^{...}{'}){\leq}2$. This process takes finite times since the rank of a density matrix is a positive integer and $\rank(\rho)>\rank(\rho')>\rank(\rho'')>\ldots$. At last, since $\rank(\rho{'}^{...}{'}){\leq}2$, $\rho{'}^{...}{'}$ can be decomposed as the equations in Case 1 or Case 2. Then let $\nu_0=\mu_0$, $\nu_1=\mu_1$, $\nu_2=\mu_2\mu'$, $\nu_3=\mu_3\mu'$, $\nu_4=\mu_4\mu'\mu''$, $\nu_5=\mu_5\mu'\mu''$ and so on.

Considering Eq.(\ref{lemma_2_1}) and Eq.(\ref{lemma_3_6}), Eq.(\ref{lemma_3_1}) holds since

\begin{align}
1=&\mu_0+\mu_1+\mu'\nonumber\\
=&\mu_0+\mu_1+\mu'(\mu_2+\mu_3+\mu'')\nonumber\\
=&\mu_0+\mu_1+\mu'\mu_2+\mu'\mu_3+\mu'\mu''(\mu_4+\mu_5+\mu''')\nonumber\\
=&\nu_0+\nu_1+\nu_2+\nu_3+\nu_4+\nu_5+\ldots\label{lemma_3_asistant_1}
\end{align}

Since Eq.(\ref{lemma_2_2}) and Eq.(\ref{lemma_3_7}), Eq.(\ref{lemma_3_2}) can be derived using the method similar to Eq.(\ref{lemma_3_asistant_1}). Eq.(\ref{lemma_3_3}) holds clearly. \\

We notice that
\begin{align}
\tr(\rho'M)=&\tr(\frac{\rho-\mu_1\rho_1+\mu_2\rho_2}{\mu'}M)\nonumber\\
=&\frac{1}{\mu'}[\tr(\rho{M})-\mu_1\tr(\rho_1{M})-\mu_2\tr(\rho_2{M})]\nonumber\\
=&\frac{1}{\mu'}[1-\mu_1-\mu_2]\tr(\rho{M})\nonumber\\
=&\tr(\rho{M})\label{VV}
\end{align}
Similar to Eq.(\ref{VV}), it is easy to obtain that $\tr(\rho{M})=\tr(\rho'{M})=\tr(\rho''{M})=\ldots$. Then $\tr(\rho{M})=\tr(\rho_1{M})=\tr(\rho_2{M})=\tr(\rho_3{M})=\tr(\rho_4{M})=\ldots$, since Eq.(\ref{lemma_2_5}) and Eq.(\ref{lemma_3_10}). It follows that Eq.(\ref{lemma_3_4}) holds.

\end{proof}


\begin{theorem}
Given the value of a linear dimension witness $\sum_{x=1}^n\sum_{y=1}^l\alpha_{xy}\tr(\rho_xM_y)=w_d$, the minimum value of the Von Neummann entropy $S(\rho)$ where $\rho=({\rho_1+\ldots+\rho_n})/n$ can be obtained when $\rho_k(1{\leq}k{\leq}n)$ are all rank-1 and in $\mathbb{C}^n$.
\end{theorem}


\begin{proof}

Let $M^{(x)}=\sum_{y=1}^{l}\alpha_{xy}M_{y}$, then the dimension witness is written as

\begin{align}
\sum_{x=1}^n\tr(\rho_xM^{(x)})=w_d \label{theorem_1}
\end{align}

Using Lemma 3, for $\rho_x$ and $M^{(x)}$, there exist density matrices $\{\rho_{x,0},\ldots,\rho_{x,s_x-1}\}$ and positive real numbers $\{\nu_{x,0},\ldots,\nu_{x,s_x-1}\}$, subject to
\begin{align}
&\sum_{k_x=0}^{s_x-1}\nu_{x,k_x}=1 \\
&\sum_{k_x=0}^{s_x-1}\nu_{x,k_x}\rho_{x,k_x}=\rho \\
&\rank(\rho_{x,k_x})=1 \text{ while } 0{\leq}k_x{\leq}{s_x-1} \\
&\tr(\rho_{x,k_x}{M^{(x)}})=\tr(\rho_x{M^{(x)}}) \text{ while } 0{\leq}k_x{\leq}{s_x-1}\label{WWWW}
\end{align}

Then
\begin{align}
&\sum_{k_1=0}^{s_1-1}\sum_{k_2=0}^{s_2-1}\ldots\sum_{k_n=0}^{s_n-1}\nu_{1,k_1}\nu_{2,k_2}\ldots\nu_{n,k_n}\nonumber\\
=&(\sum_{k_1=0}^{s_1-1}\nu_{1,k_1})(\sum_{k_2=0}^{s_2-1}\nu_{2,k_2})\ldots(\sum_{k_n=0}^{s_n-1}\nu_{n,k_n})\nonumber\\
=&1
\end{align}

Furthermore
\begin{align}
\rho=&\frac{1}{n}({\rho_1+\ldots+\rho_n}) \nonumber\\
=&\frac{1}{n}[(\sum_{k_1=0}^{s_1-1}\nu_{1,k_1}\rho_{1,k_1})+(\sum_{k_2=0}^{s_2-1}\nu_{2,k_2}\rho_{2,k_2})+\ldots+(\sum_{k_n=0}^{s_n-1}\nu_{n,k_n}\rho_{n,k_n})]\nonumber\\
=&\sum_{k_1=1}^{s_1}\sum_{k_2=1}^{s_2}\ldots\sum_{k_n=1}^{s_n}\nu_{1,k_1}\nu_{2,k_2}\ldots\nu_{n,k_n}[\frac{1}{n}({\rho_{1,k_1}+\ldots+\rho_{n,k_n}})]\label{theorem_2}
\end{align}

Since $({\rho_{1,k_1}+\ldots+\rho_{n,k_n}})/n$ is also a density matrix and Eq.(2.2) in the Page 237 of Ref.\cite{Wehrl1978}, it follows that
\begin{align}
S(\rho)=&S(\sum_{k_1=0}^{s_1-1}\sum_{k_2=0}^{s_2-1}\ldots\sum_{k_n=0}^{s_n-1}\nu_{1,k_1}\nu_{2,k_2}\ldots\nu_{n,k_n}[\frac{1}{n}({\rho_{1,k_1}+\ldots+\rho_{n,k_n}})])\nonumber\\
{\geq}&\sum_{k_1=0}^{s_1-1}\sum_{k_2=0}^{s_2-1}\ldots\sum_{k_n=0}^{s_n-1}\nu_{1,k_1}\nu_{2,k_2}\ldots\nu_{n,k_n}S(\frac{{\rho_{1,k_1}+\ldots+\rho_{n,k_n}}}{n})\nonumber\\
{\geq}&\min_{\substack{0{\leq}k_1{\leq}s_1-1 \\ \vdots \\ 0{\leq}k_n{\leq}s_n-1}}S(\frac{{\rho_{1,k_1}+\ldots+\rho_{n,k_n}}}{n})\nonumber\\
=&S(\frac{{\rho_{1,t_1}+\ldots+\rho_{n,t_n}}}{n})\label{theorem_3}
\end{align}

Where $t_1\in\{0,\ldots,s_1-1\}$, $\ldots$, $t_n\in\{0,\ldots,s_n-1\}$.

On the other hand, while $\rho_1=\rho_{1,t_1}$, $\ldots$, $\rho_n=\rho_{n,t_n}$, the equation of the linear dimension witness $\sum_{x=1}^n\tr(\rho_xM^{(x)})=w_d$ holds since Eq.(\ref{WWWW}).

Then considering that $\{\rho_{x,t_x}\}$ are rank-1 density matrices, they are written as
\begin{align}
\rho_{x,t_x}=|\Psi_x{\rangle}{\langle}\Psi_x| \text{ while } 1{\leq}x{\leq}n\label{theorem_4}
\end{align}

Let
\begin{align}
|\Psi_1'{\rangle}=&|\Psi_1{\rangle}\\
|\Psi_2'{\rangle}=&\frac{|\Psi_2{\rangle}-{\langle}\Psi_1'|\Psi_2{\rangle}|\Psi_1'{\rangle}}{\||\Psi_2{\rangle}-{\langle}\Psi_1'|\Psi_2{\rangle}|\Psi_1'{\rangle}\|}\\
|\Psi_3'{\rangle}=&\frac{|\Psi_3{\rangle}-{\langle}\Psi_1'|\Psi_3{\rangle}|\Psi_1'{\rangle}-{\langle}\Psi_2'|\Psi_3{\rangle}|\Psi_2'{\rangle}}{\||\Psi_3{\rangle}-{\langle}\Psi_1'|\Psi_3{\rangle}|\Psi_1'{\rangle}-{\langle}\Psi_2'|\Psi_3{\rangle}|\Psi_2'{\rangle}\|}\\
\ldots=&\ldots\nonumber\\
|\Psi_n'{\rangle}=&\frac{|\Psi_n{\rangle}-{\langle}\Psi_1'|\Psi_n{\rangle}|\Psi_1'{\rangle}-{\langle}\Psi_2'|\Psi_n{\rangle}|\Psi_2'{\rangle}-\ldots-{\langle}\Psi_{n-1}'|\Psi_n{\rangle}|\Psi_{n-1}'{\rangle}}{\||\Psi_n{\rangle}-{\langle}\Psi_1'|\Psi_n{\rangle}|\Psi_1'{\rangle}-{\langle}\Psi_2'|\Psi_n{\rangle}|\Psi_2'{\rangle}-\ldots-{\langle}\Psi_{n-1}'|\Psi_n{\rangle}|\Psi_{n-1}'{\rangle}\|}\label{theorem_5}
\end{align}

Then $\{|\Psi_1'{\rangle},\ldots,|\Psi_n'{\rangle}\}$ are orthogonal pairwise and $\{|\Psi_1{\rangle},\ldots,|\Psi_n{\rangle}\}$ are in the space $\Sigma=\spann\{|\Psi_1'{\rangle},\ldots,|\Psi_n'{\rangle}\}$.
Since Eq.(\ref{theorem_4}), $\rho_{1,t_1},\rho_{2,t_2},\ldots,\rho_{n,t_n}$ are all in the space $\Sigma$. Since dim$(\Sigma){\leq}n$, $\Sigma$ is included in $\mathbb{C}^n$.

Hence, given the value of a linear dimension witness $\sum_{x=1}^n\tr(\rho_xM^{(x)})=w_d$, for any density matrices $\{\rho_1,\ldots,\rho_n\}$, there exist density matrices $\{\rho_{1,t_1},\ldots,\rho_{1,t_n}\}$, subject to
\begin{align}
&\rank(\rho_{x,t_x})=1 \text{ while } 1{\leq}x{\leq}n\\
&\rho_{x,t_x}\in\mathbb{C}^n \text{ while } 1{\leq}x{\leq}n\\
&\sum_{x=1}^n\tr(\rho_{x,t_x}M^{(x)})=w_d\\
&S(\frac{\rho_{1,t_1}+\ldots+\rho_{n,t_n}}{n}){\leq}S(\frac{\rho_1+\ldots+\rho_n}{n})
\end{align}

Hence, given the value of a linear dimension witness $\sum_{x=1}^n\sum_{y=1}^l\alpha_{xy}\tr(\rho_xM_y)=w_d$, the minimal value of the Von Neummann entropy $S(\rho)$ where $\rho=({\rho_1+\ldots+\rho_n})/n$ is equal to
\begin{align}
&{\min}S(\frac{\rho_1+\ldots+\rho_n}{n})\nonumber\\
\text{ s.t. }&\nonumber\\
&\sum_{x=1}^n\sum_{y=1}^l\alpha_{xy}\tr(\rho_xM_y)=w_d\nonumber\\
&\rank(\rho_{x})=1 \text{ while } 1{\leq}x{\leq}n\nonumber\\
&\rho_x\in\mathbb{C}^n\text{ while }1{\leq}x{\leq}n\nonumber\\
\end{align}

\end{proof}

\newpage

\subsection{B. Details about the maximal differences between minimal values of $H(M)$ and $S(\rho)$ for $I_3$, $I_4$, and $R_4$.}

The states $\rho_x$, the measurements $M_y$, the deterministic expectation values $E_{m,y}^{(\lambda)}$, the deterministic probability distribution $P_{m,x}^{(\lambda)}$ and the efficiency matrix $A_{x,y}$ are written as

\begin{align}
\rho_x=&|\psi_x{\rangle}{\langle}\psi_x|\\
M_y=&1-2|m_y{\rangle}{\langle}m_y|\\
E_{m,y}^{(\lambda)}=&\begin{bmatrix}   E(m=0,y=1,\lambda) & \ldots &  E(m=n-1,y=1,\lambda) \\ & \vdots & \\   E(m=0,y=l,\lambda) & \ldots &  E(m=n-1,y=l,\lambda) \end{bmatrix}\\
P_{m,x}^{(\lambda)}=&\begin{bmatrix}   P(m=0|x=1,\lambda) & \ldots &  P(m=0|x=n,\lambda) \\ & \vdots & \\   P(m=n-1|x=1,\lambda) & \ldots &  P(m=n-1|x=n,\lambda) \end{bmatrix}\\
A_{xy}=&\begin{bmatrix} \alpha_{x=1,y=1} & \ldots & \alpha_{x=1,y=l} \\ & \vdots & \\ \alpha_{x=n,y=1} & \ldots & \alpha_{x=n,y=l} \end{bmatrix}
\end{align}

Here we notice that $P_{m,x}^{(\lambda)}$ has $n$ rows and $E_{m,y}^{(\lambda)}$ has $n$ columns, since the message $M$ with dimension $n$ is proved to be sufficient in Sec.III of Supplementary Material of Ref.\cite{chaves2015prl}. While $\rank\{P_{m,x}^{(\lambda)}\}$ is less than $n$, the dimension witness of $w_d$ can be obtained by a system with dimension lower than $n$ .

For the quantum entropy,
\begin{align}
S(\rho)=-\tr(\rho\log_2\rho) \text{ , where }\rho=\sum_{x=1}^n\rho_x/n\label{detail-1}
\end{align}

For the classical entropy,
\begin{align}
H(M)=-\sum_{m=0}^{n-1}p_m\log_2p_m \text{ ,  where }p_m=\sum_{x=1}^n\sum_\lambda{P(m|x,\lambda)}q_\lambda/n\label{detail-2}
\end{align}

For the quantum dimension witness,
\begin{align}
w_d^{(q)}=\sum_{x=1}^n\sum_{y=1}^l\alpha_{xy}\tr(\rho_xM_y)\label{detail-3}
\end{align}

For the classical dimension witness,
\begin{align}
w_d^{(c)}=\sum_{x=1}^n\sum_{y=1}^l\alpha_{xy}\sum_{m=0}^{n-1}\sum_{\lambda}E(y,m,\lambda)P(m|x,\lambda)q_\lambda=\sum_\lambda\tr\{A_{xy}E_{m,y}^{(\lambda)}P_{m,x}^{(\lambda)}\}q_\lambda\label{detail-4}
\end{align}

While accessing the values shown in TABLE.I of the main text, the details about the states $|\psi_x{\rangle}$, the projection states $|m_y{\rangle}$, the deterministic expectation values $E_{m,y}^{(\lambda)}$, the deterministic probability distribution $P_{m,x}^{(\lambda)}$ and the probability of strategies $q_\lambda$ are shown below.

\ \\

\noindent{\bf For the case of $I_3$}

\ \\
The efficiency matrix is
\begin{align}
A_{xy}=&\begin{bmatrix} 1 & 1 \\ 1 & -1 \\ -1 & 0 \end{bmatrix}\label{pre_i3}
\end{align}

The quantum states are
\begin{align}
|\psi_1{\rangle}&=(1,0,0)\label{i3-pre-begin}\\
|\psi_2{\rangle}&=(0.7972,0.6037,0)\\
|\psi_3{\rangle}&=(0.6511,-0.7590,0)\label{i3-pre-end}
\end{align}

The projection states are
\begin{align}
|m_1{\rangle}&=(0.4531,-0.8914,0)\label{i3-mea-begin}\\
|m_2{\rangle}&=(0.4451,0.8955,0)\label{i3-mea-end}
\end{align}

There are two classical strategies $\lambda_1$ and $\lambda_2$, their probabilities are
\begin{align}
q_{\lambda_1}&=0.3111\\
q_{\lambda_2}&=0.6889
\end{align}

The deterministic expectation values are
\begin{align}
E_{m,y}^{(\lambda_1)}&=E_{m,y}^{(\lambda_2)}=\begin{bmatrix} 1 & 1 &-1 \\ 1 & -1 & 1 \end{bmatrix}\label{i3-deter}
\end{align}

The deterministic probability distributions are
\begin{align}
P_{m,x}^{(\lambda_1)}&=\begin{bmatrix} 1 & 0 &0 \\ 0 & 1 & 0 \\ 0 & 0 & 1 \end{bmatrix}\label{pre-i3-before}\\
P_{m,x}^{(\lambda_2)}&=\begin{bmatrix} 1 & 0 &1 \\ 0 & 1 & 0 \\ 0 & 0 & 0 \end{bmatrix}\label{pre-i3}
\end{align}

Substitute Eq.(\ref{pre_i3})$\sim$(\ref{pre-i3}) into Eq.(\ref{detail-1})$\sim$(\ref{detail-4}),
\begin{align}
w_d^{(c)}=w_d^{(q)}=&3.622\nonumber\\
H(M)=&1.334\text{ bit}\nonumber\\
S(\rho)=&0.897\text{ bit}
\end{align}

\ \\

\noindent{\bf For the case of $I_4$}

\ \\

The efficiency matrix is
\begin{align}
A_{xy}=&\begin{bmatrix} 1 & 1 & 1 \\ 1 & 1 & -1 \\  1 & -1 & 0 \\ -1 & 0 & 0 \end{bmatrix}\label{pre_i4}
\end{align}

The quantum states are
\begin{align}
|\psi_1{\rangle}&=(1,0,0,0)\label{i4-pre-begin}\\
|\psi_2{\rangle}&=(0.8323,0.5543,0,0)\\
|\psi_3{\rangle}&=(0.3108,0.9505,0,0)\\
|\psi_4{\rangle}&=(0.7623,0.5247,0.2148,0.3121)\label{i4-pre-end}
\end{align}

The projection states are
\begin{align}
|m_1{\rangle}&=(0.1692,0.1164,0.5549,0.8062)\label{i4-mea-begin}\\
|m_2{\rangle}&=(0.0750,-0.9972,0,0)\\
|m_3{\rangle}&=(0.4721,0.8816,0,0)\label{i4-mea-end}
\end{align}

There are two classical strategies $\lambda_1$ and $\lambda_2$, their probabilities are
\begin{align}
q_{\lambda_1}&=0.3802\\
q_{\lambda_2}&=0.6198
\end{align}

The deterministic expectation values are
\begin{align}
E_{m,y}^{(\lambda_1)}&=E_{m,y}^{(\lambda_2)}=\begin{bmatrix} 1 & 1 & 1 &-1 \\ 1 & 1 & -1 & 1 \\ 1 & -1 & 1 & 1 \end{bmatrix}\label{i4-deter}
\end{align}

The deterministic probability distributions are
\begin{align}
P_{m,x}^{(\lambda_1)}&=\begin{bmatrix} 1 & 0 &0 &1 \\ 0 & 1 & 0 & 0\\ 0 & 0 & 1 & 0 \\ 0 & 0 & 0 & 0 \end{bmatrix}\label{pre-i4-before}\\
P_{m,x}^{(\lambda_2)}&=\begin{bmatrix} 1 & 0 &1 & 1 \\ 0 & 1 & 0 & 0\\ 0 & 0 & 0 & 0 \\ 0 & 0 & 0 & 0 \end{bmatrix}\label{pre-i4}
\end{align}

Substitute Eq.(\ref{pre_i4})$\sim$(\ref{pre-i4}) into Eq.(\ref{detail-1})$\sim$(\ref{detail-4})
\begin{align}
w_d^{(c)}=w_d^{(q)}=&5.760\nonumber\\
H(M)=&1.223\text{ bit}\nonumber\\
S(\rho)=&0.829\text{ bit}
\end{align}


\ \\

\noindent{\bf For the case of $R_4$}

\ \\

The efficiency matrix is
\begin{align}
A_{xy}=&\begin{bmatrix} 1 & 1 \\ 1 & -1 \\ -1 & 1 \\ -1 & -1 \end{bmatrix}\label{pre_r4}
\end{align}

The quantum states are
\begin{align}
|\psi_1{\rangle}&=(1,0,0,0)\label{r4-pre-begin}\\
|\psi_2{\rangle}&=(0.7588,0.2363-0.6070i,0,0)\\
|\psi_3{\rangle}&=(0.7588,0.2363+0.6070i,0,0)\\
|\psi_4{\rangle}&=(0.3893,0.9211,0,0)\label{r4-pre-end}
\end{align}

The projection states are
\begin{align}
|m_1{\rangle}&=(0.1515-0.3891i,0.9087,0,0)\label{r4-mea-begin}\\
|m_2{\rangle}&=(0.1515+0.3891i,0.9087,0,0)\label{r4-mea-end}
\end{align}

There are two classical strategies $\lambda_1$ and $\lambda_2$, their probabilities are
\begin{align}
q_{\lambda_1}&=0.6056\\
q_{\lambda_2}&=0.3944
\end{align}

The deterministic expectation values are
\begin{align}
E_{m,y}^{(\lambda_1)}=E_{m,y}^{(\lambda_2)}&=\begin{bmatrix} 1 & 1 & -1 &-1 \\ 1 & -1 & 1 & -1 \end{bmatrix}\label{r4-deter}
\end{align}

The deterministic probability distributions are
\begin{align}
P_{m,x}^{(\lambda_1)}&=\begin{bmatrix} 1 & 0 &1 &0 \\ 0 & 1 & 0 & 0\\ 0 & 0 & 0 & 0 \\ 0 & 0 & 0 & 1 \end{bmatrix}\label{pre-r4-before}\\
P_{m,x}^{(\lambda_2)}&=\begin{bmatrix} 1 & 1 &1 & 0 \\ 0 & 0 & 0 & 0\\ 0 & 0 & 0 & 0 \\ 0 & 0 & 0 & 1 \end{bmatrix}\label{pre-r4}
\end{align}

Substitute Eq.(\ref{pre_r4})$\sim$(\ref{pre-r4}) into Eq.(\ref{detail-1})$\sim$(\ref{detail-4})
\begin{align}
w_d^{(c)}=w_d^{(q)}=&5.211\nonumber\\
H(M)=&1.356\text{ bit}\nonumber\\
S(\rho)=&0.888\text{ bit}
\end{align}

\newpage

\subsection{C. The rotation angles of HWPs and QWPs}
\ \\
\noindent{\bf The preparation of quantum states}
\ \\

Following Eq.(7) in the main text and Eq.(\ref{i3-pre-begin})$\sim$Eq.(\ref{i3-pre-end}),
\begin{table}[!htbp]
\centering
\caption{The rotation angles of HWPs and QWP in the state preparator for quantum states in the case of $I_3$.}\label{t-1}
  \begin{tabular}{ cccc}
    \hline \hline
          & ${h_s^{(p)}}$ & $q_s^{(p)}$ & $h_i^{(p)}$ \\
    \hline    
    $|\psi_1{\rangle}$ & 0$^{\circ}$ & 0$^{\circ}$ & 0$^{\circ}$  \\
    $|\psi_2{\rangle}$ & 18.57$^{\circ}$ & 37.14$^{\circ}$ & 0$^{\circ}$  \\
    $|\psi_3{\rangle}$ & -24.69$^{\circ}$ & -49.38$^{\circ}$ & 0$^{\circ}$  \\
    \hline
    \hline
  \end{tabular}
\end{table}

Following Eq.(7) in the main text and Eq.(\ref{i4-pre-begin})$\sim$Eq.(\ref{i4-pre-end}),

\begin{table}[!htbp]
\centering
\caption{The rotation angles of HWPs and QWP in the state preparator for quantum states in the case of $I_4$.}\label{t-2}
  \begin{tabular}{ cccc}
    \hline \hline
          & ${h_s^{(p)}}$ & $q_s^{(p)}$ & $h_i^{(p)}$ \\
    \hline    
    $|\psi_1{\rangle}$ & 0$^{\circ}$ & 0$^{\circ}$ & 0$^{\circ}$  \\
    $|\psi_2{\rangle}$ & 16.83$^{\circ}$ & 33.66$^{\circ}$ & 0$^{\circ}$  \\
    $|\psi_3{\rangle}$ & 35.95$^{\circ}$ & 71.89$^{\circ}$ & 0$^{\circ}$  \\
    $|\psi_4{\rangle}$ & 17.27$^{\circ}$ & 34.54$^{\circ}$ & 11.13$^{\circ}$  \\
    \hline
    \hline
  \end{tabular}
\end{table}

Following Eq.(7) in the main text and Eq.(\ref{r4-pre-begin})$\sim$Eq.(\ref{r4-pre-end}),

\begin{table}[!htbp]
\centering
\caption{The rotation angles of HWPs and QWP in the state preparator for quantum states in the case of $R_4$.}\label{t-3}
  \begin{tabular}{ cccc}
    \hline \hline
          & ${h_s^{(p)}}$ & $q_s^{(p)}$ & $h_i^{(p)}$ \\
    \hline    
    $|\psi_1{\rangle}$ & 0$^{\circ}$ & 0$^{\circ}$ & 0$^{\circ}$  \\
    $|\psi_2{\rangle}$ & 33.55$^{\circ}$ & 33.55$^{\circ}$ & 0$^{\circ}$  \\
    $|\psi_3{\rangle}$ & 0$^{\circ}$ & 33.55$^{\circ}$ & 0$^{\circ}$  \\
    $|\psi_4{\rangle}$ & 33.55$^{\circ}$ & 67.09$^{\circ}$ & 0$^{\circ}$  \\
    \hline
    \hline
  \end{tabular}
\end{table}

\ \\
\noindent{\bf The preparation of classical states}
\ \\

Following Eq.(7) in the main text and Eq.(\ref{pre-i3-before})$\sim$Eq.(\ref{pre-i3}),

\begin{table}[!htbp]
\centering
\caption{The rotation angles of HWPs and QWP in the state preparator for classical states of strategy $\lambda_1$ in the case of $I_3$.}\label{t-4}
  \begin{tabular}{ cccc}
    \hline \hline
          & ${h_s^{(p)}}$ & $q_s^{(p)}$ & $h_i^{(p)}$ \\
    \hline    
    $\text{State 1}$ & 0$^{\circ}$ & 0$^{\circ}$ & 0$^{\circ}$  \\
    $\text{State 2}$ & 45$^{\circ}$ & 90$^{\circ}$ & 0$^{\circ}$  \\
    $\text{State 3}$ & 45$^{\circ}$ & 90$^{\circ}$ & 45$^{\circ}$  \\
    \hline
    \hline
  \end{tabular}
\end{table}

\begin{table}[!htbp]
\centering
\caption{The rotation angles of HWPs and QWP in the state preparator for classical states of strategy $\lambda_2$ in the case of $I_3$.}\label{t-4}
  \begin{tabular}{ cccc}
    \hline \hline
          & ${h_s^{(p)}}$ & $q_s^{(p)}$ & $h_i^{(p)}$ \\
    \hline    
    $\text{State 1}$ & 0$^{\circ}$ & 0$^{\circ}$ & 0$^{\circ}$  \\
    $\text{State 2}$ & 45$^{\circ}$ & 90$^{\circ}$ & 0$^{\circ}$  \\
    $\text{State 3}$ & 0$^{\circ}$ & 0$^{\circ}$ & 0$^{\circ}$  \\
    \hline
    \hline
  \end{tabular}
\end{table}

\newpage

Following Eq.(7) in the main text and Eq.(\ref{pre-i4-before})$\sim$Eq.(\ref{pre-i4}),

\begin{table}[!htbp]
\centering
\caption{The rotation angles of HWPs and QWP in the state preparator for classical states of strategy $\lambda_1$ in the case of $I_4$.}\label{t-5}
  \begin{tabular}{ cccc}
    \hline \hline
          & ${h_s^{(p)}}$ & $q_s^{(p)}$ & $h_i^{(p)}$ \\
    \hline    
    $\text{State 1}$ & 0$^{\circ}$ & 0$^{\circ}$ & 0$^{\circ}$  \\
    $\text{State 2}$ & 45$^{\circ}$ & 90$^{\circ}$ & 0$^{\circ}$  \\
    $\text{State 3 }$ & 45$^{\circ}$ & 90$^{\circ}$ & 45$^{\circ}$  \\
    $\text{State 4}$ & 0$^{\circ}$ & 0$^{\circ}$ & 0$^{\circ}$  \\
    \hline
    \hline
  \end{tabular}
\end{table}

\begin{table}[!htbp]
\centering
\caption{The rotation angles of HWPs and QWP in the state preparator for classical states of strategy $\lambda_2$ in the case of $I_4$.}\label{t-5}
  \begin{tabular}{ cccc}
    \hline \hline
          & ${h_s^{(p)}}$ & $q_s^{(p)}$ & $h_i^{(p)}$ \\
    \hline    
    $\text{State 1}$ & 0$^{\circ}$ & 0$^{\circ}$ & 0$^{\circ}$  \\
    $\text{State 2}$ & 45$^{\circ}$ & 90$^{\circ}$ & 0$^{\circ}$  \\
    $\text{State 3}$ & 0$^{\circ}$ & 0$^{\circ}$ & 0$^{\circ}$  \\
    $\text{State 4}$ & 0$^{\circ}$ & 0$^{\circ}$ & 0$^{\circ}$  \\
    \hline
    \hline
  \end{tabular}
\end{table}

Following Eq.(7) in the main text and Eq.(\ref{pre-r4-before})$\sim$Eq.(\ref{pre-r4}),

\begin{table}[!htbp]
\centering
\caption{The rotation angles of HWPs and QWP in the state preparator for classical states of strategy $\lambda_1$ in the case of $R_4$.}\label{t-6}
  \begin{tabular}{ cccc}
    \hline \hline
          & ${h_s^{(p)}}$ & $q_s^{(p)}$ & $h_i^{(p)}$ \\
    \hline    
    $\text{State 1}$ & 0$^{\circ}$ & 0$^{\circ}$ & 0$^{\circ}$  \\
    $\text{State 2}$ & 45$^{\circ}$ & 90$^{\circ}$ & 0$^{\circ}$  \\
    $\text{State 3}$ & 0$^{\circ}$ & 0$^{\circ}$ & 0$^{\circ}$  \\
    $\text{State 4}$ & 0$^{\circ}$ & 0$^{\circ}$ & 45$^{\circ}$  \\
    \hline
    \hline
  \end{tabular}
\end{table}

\begin{table}[!htbp]
\centering
\caption{The rotation angles of HWPs and QWP in the state preparator for classical states of strategy $\lambda_2$ in the case of $R_4$.}\label{t-6}
  \begin{tabular}{ cccc}
    \hline \hline
          & ${h_s^{(p)}}$ & $q_s^{(p)}$ & $h_i^{(p)}$ \\
    \hline    
    $\text{State 1}$ & 0$^{\circ}$ & 0$^{\circ}$ & 0$^{\circ}$  \\
    $\text{State 2}$ & 0$^{\circ}$ & 0$^{\circ}$ & 0$^{\circ}$  \\
    $\text{State 3}$ & 0$^{\circ}$ & 0$^{\circ}$ & 0$^{\circ}$  \\
    $\text{State 4}$ & 0$^{\circ}$ & 0$^{\circ}$ & 45$^{\circ}$  \\
    \hline
    \hline
  \end{tabular}
\end{table}

\ \\
\noindent{\bf The detection of quantum dimension witness}
\ \\

The expectations of detect-events for the quantum dimension witness in the case of $I_3$, $I_4$, and $R_4$ are
\begin{align}
E=&\frac{-D_{a,b}+D_{c,b}+D_{c,d}}{D_{a,b}+D_{c,b}+D_{c,d}} \text{     for $I_3$}\\
E=&\frac{-D_{a,b}+D_{c,b}+D_{c,d}+D_{a,d}}{D_{a,b}+D_{c,b}+D_{c,d}+D_{a,d}} \text{     for $I_4$}\\
E=&\frac{-D_{a,b}+D_{c,b}+D_{c,d}+D_{a,d}}{D_{a,b}+D_{c,b}+D_{c,d}+D_{a,d}} \text{     for $R_4$}
\end{align}

\newpage

Following Eq.(8) in the main text and Eq.(\ref{i3-mea-begin})$\sim$Eq.(\ref{i3-mea-end}),

\begin{table}[!htbp]
\centering
\caption{The rotation angles of HWPs and QWPs in the measurement device for detection of quantum states in the case of $I_3$.}\label{t-7}
  \begin{tabular}{ ccccc}
    \hline \hline
          & ${h_s^{(m)}}$ & $q_s^{(m)}$ & $h_i^{(m)}$ & $q_i^{(m)}$ \\
    \hline    
    $|m_1{\rangle}$ & -31.53$^{\circ}$ & -63.06$^{\circ}$ & 0$^{\circ}$ & 0$^{\circ}$  \\
    $|m_2{\rangle}$ & 31.79$^{\circ}$ & 63.57$^{\circ}$ & 0$^{\circ}$ & 0$^{\circ}$  \\
    \hline
    \hline
  \end{tabular}
\end{table}

Following Eq.(8) in the main text and Eq.(\ref{i4-mea-begin})$\sim$Eq.(\ref{i4-mea-end}),

\begin{table}[!htbp]
\centering
\caption{The rotation angles of HWPs and QWPs in the measurement device for detection of quantum states in the case of $I_4$.}\label{t-8}
  \begin{tabular}{ ccccc}
    \hline \hline
          & ${h_s^{(p)}}$ & $q_s^{(p)}$ & $h_i^{(p)}$ & $q_i^{(m)}$ \\
    \hline    
    $|m_1{\rangle}$ & 17.26$^{\circ}$ & 34.53$^{\circ}$ & 39.07$^{\circ}$ & 78.15$^{\circ}$  \\
    $|m_2{\rangle}$ & -42.85$^{\circ}$ & -85.70$^{\circ}$ & 0$^{\circ}$ & 0$^{\circ}$  \\
    $|m_3{\rangle}$ & 30.92$^{\circ}$ & 61.84$^{\circ}$ & 0$^{\circ}$ & 0$^{\circ}$  \\
    \hline
    \hline
  \end{tabular}
\end{table}

Following Eq.(8) in the main text and Eq.(\ref{r4-mea-begin})$\sim$Eq.(\ref{r4-mea-end}),

\begin{table}[!htbp]
\centering
\caption{The rotation angles of HWPs and QWPs in the measurement device for detection of quantum states in the case of $R_4$.}\label{t-9}
  \begin{tabular}{ ccccc}
    \hline \hline
          & ${h_s^{(p)}}$ & $q_s^{(p)}$ & $h_i^{(p)}$ & $q_i^{(m)}$ \\
    \hline    
    $|m_1{\rangle}$ & 50.52$^{\circ}$ & 78.54$^{\circ}$ & 0$^{\circ}$ & 0$^{\circ}$  \\
    $|m_2{\rangle}$ & 28.02$^{\circ}$ & 78.54$^{\circ}$ & 0$^{\circ}$ & 0$^{\circ}$  \\
    \hline
    \hline
  \end{tabular}
\end{table}

\ \\
\noindent{\bf The detection of classical dimension witness}
\ \\

Following Eq.(\ref{i3-deter}), the expectations of detect-events for the classical dimension witness in the case of $I_3$ are
\begin{align}
E=&\frac{D_{a,b}+D_{c,b}-D_{c,d}}{D_{a,b}+D_{c,b}+D_{c,d}} \text{     for $M_1$}\\
E=&\frac{D_{a,b}-D_{c,b}+D_{c,d}}{D_{a,b}+D_{c,b}+D_{c,d}} \text{     for $M_2$}
\end{align}

Following Eq.(\ref{i4-deter}), the expectations of detect-events for the classical dimension witness in the case of $I_4$ are
\begin{align}
E=&\frac{D_{a,b}+D_{c,b}+D_{c,d}-D_{a,d}}{D_{a,b}+D_{c,b}+D_{c,d}+D_{a,d}} \text{     for $M_1$}\\
E=&\frac{D_{a,b}+D_{c,b}-D_{c,d}+D_{a,d}}{D_{a,b}+D_{c,b}+D_{c,d}+D_{a,d}} \text{     for $M_2$}\\
E=&\frac{D_{a,b}-D_{c,b}+D_{c,d}+D_{a,d}}{D_{a,b}+D_{c,b}+D_{c,d}+D_{a,d}} \text{     for $M_3$}
\end{align}

Following Eq.(\ref{r4-deter}), the expectations of detect-events for the classical dimension witness in the case of $R_4$ are
\begin{align}
E=&\frac{D_{a,b}+D_{c,b}-D_{c,d}-D_{a,d}}{D_{a,b}+D_{c,b}+D_{c,d}+D_{a,d}} \text{     for $M_1$}\\
E=&\frac{D_{a,b}-D_{c,b}+D_{c,d}-D_{a,d}}{D_{a,b}+D_{c,b}+D_{c,d}+D_{a,d}} \text{     for $M_2$}
\end{align}

The rotation angles of HWPs and QWPs in the measurement device for classical states in the case of $I_3$, $I_4$, and $R_4$ are all $0^{\circ}$.

\newpage

\ \\
\noindent{\bf The detection of quantum entropy}
\ \\

In quantum state tomography, for the reconstruction of a $s$ order density matrix, $s^2$ projection states $|\nu_j{\rangle}$ are utilized where their projective operators are linearly independent. These projection states are realized by rotating angles of $h_s^{(m)}$,$q_s^{(m)}$,$h_i^{(m)}$ and $q_i^{(m)}$ following the Eq.(8) in the main text. The detect-events $D_{a,b}(|\nu_j{\rangle})$ which represents the coincidence number between port 'a' and 'b' while the projection state is $|\nu_j{\rangle}$ is

\begin{align}
D_{a,b}(|\nu_j{\rangle})=N{\langle}\nu_j|\rho_x|\nu_j{\rangle} \text{ while } 1{\leq}j{\leq}s^2
\end{align}

N is a constant. Since $\rho$ has $s^2-1$ independent variables, it can be linear reconstructed by $D_{a,b}(|\nu_j{\rangle})$

\begin{align}
\rho_x=\frac{\sum_{j=1}^{s^2}M_jD_{a,b}(|\nu_j{\rangle})}{\sum_{j=1}^{s}D_{a,b}(|\nu_j{\rangle})}\label{eeee}
\end{align}

$M_j$($1{\leq}j{\leq}s^2$) are the matrixes which depend on $|\nu_j{\rangle}$. To keep the positive semi-definiteness of $\rho_x$, the maximum likelihood estimation\cite{james2001pra} is used.

\ \\
For the case of $I_3$, $s=3$ and each of $\rho_1$, $\rho_2$ and $\rho_3$ is a 3 order density matrix. We reconstruct $\rho_1$, $\rho_2$ and $\rho_3$ and then obtain the average state as

\begin{align}
\rho=\frac{1}{3}(\rho_1+\rho_2+\rho_3)
\end{align}

\begin{table}[!htbp]
\centering
\caption{The rotation angles of HWPs and QWPs in the measurement device in the case of $I_3$.}\label{t-14}
  \begin{tabular}{ ccccc}
      \hline \hline
          & ${h_s^{(m)}}$ & $q_s^{(m)}$ & $h_i^{(m)}$ & $q_i^{(m)}$ \\
    \hline    
$D_{a,b}(|\nu_1{\rangle})$ &      $0^{\circ}$ &$0^{\circ}$&$0^{\circ}$&$0^{\circ}$\\
$D_{a,b}(|\nu_2{\rangle})$ &      $45^{\circ}$ &$0^{\circ}$&$0^{\circ}$&$0^{\circ}$\\
$D_{a,b}(|\nu_3{\rangle})$ &      $45^{\circ}$&$0^{\circ}$&$45^{\circ}$&$0^{\circ}$\\
$D_{a,b}(|\nu_4{\rangle})$ &      $45^{\circ}$&$0^{\circ}$&$22.5^{\circ}$&$0^{\circ}$\\
$D_{a,b}(|\nu_5{\rangle})$ &      $45^{\circ}$&$0^{\circ}$&$22.5^{\circ}$&$45^{\circ}$\\
$D_{a,b}(|\nu_6{\rangle})$ &      $22.5^{\circ}$&$45^{\circ}$&$22.5^{\circ}$&$45^{\circ}$\\
$D_{a,b}(|\nu_7{\rangle})$ &      $22.5^{\circ}$&$45^{\circ}$&$22.5^{\circ}$&$90^{\circ}$\\
$D_{a,b}(|\nu_8{\rangle})$ &      $22.5^{\circ}$&$45^{\circ}$&$0^{\circ}$&$90^{\circ}$\\
$D_{a,b}(|\nu_9{\rangle})$ &      $22.5^{\circ}$&$0^{\circ}$&$0^{\circ}$&$90^{\circ}$\\
   \hline
    \hline
  \end{tabular}
\end{table}

The matrixes $M_j$ ($1{\leq}j{\leq}9$) are

\begin{align}
&M_1=\frac{1}{2}\begin{bmatrix} 2 & -1+i & 0 \\ -1-i & 0 & 0 \\ 0 & 0 & 0 \end{bmatrix} &
&M_2=\frac{1}{2}\begin{bmatrix} 0 & -1+i & 1-i \\ -1-i & 2 & -1+i \\ 1+i & -1-i & 0 \end{bmatrix} &
&M_3=\frac{1}{2}\begin{bmatrix} 0 & 0 & -2i \\ 0 & 0 & -1+i \\ 2i & -1-i & 2 \end{bmatrix} \nonumber\\
&M_4=\begin{bmatrix} 0 & 0 & i \\ 0 & 0 & -i \\ -i & i & 0  \end{bmatrix} &
&M_5=\begin{bmatrix}  0 & 0 & -1 \\ 0 & 0 & 1 \\ -1 & 1 & 0 \end{bmatrix} &
&M_6=\begin{bmatrix} 0 & 0 & 2 \\ 0 & 0 & 0 \\ 2 & 0 & 0 \end{bmatrix}\nonumber\\
&M_7=\begin{bmatrix} 0 & 0 & 2i \\ 0 & 0 & 0 \\ -2i & 0 & 0 \end{bmatrix} &
&M_8=\begin{bmatrix} 0 & 1 & -1-i \\ 1 & 0 & 0 \\ -1+i & 0 & 0 \end{bmatrix} &
&M_9=\begin{bmatrix} 0 & -i & 0 \\ i & 0 & 0 \\ 0 & 0 & 0 \end{bmatrix}
\end{align}

\ \\
For the case of $I_4$ and $R_4$, $s=4$ and each of $\rho_1$, $\rho_2$, $\rho_3$ and $\rho_4$ is a 4 order density matrix. We reconstruct $\rho_1$, $\rho_2$, $\rho_3$ and $\rho_4$ and then obtain the average state as

\newpage

\begin{align}
\rho=\frac{1}{4}(\rho_1+\rho_2+\rho_3+\rho_4).
\end{align}

\begin{table}[!htbp]
\centering
\caption{The rotation angles of HWPs and QWPs in the measurement device in the cases of $I_4$ and $R_4$. }\label{t-15}
  \begin{tabular}{ ccccc}
    \hline \hline
          & ${h_s^{(m)}}$ & $q_s^{(m)}$ & $h_i^{(m)}$ & $q_i^{(m)}$ \\
    \hline    
$D_{a,b}(|\nu_1{\rangle})$ &      $45^{\circ}$&$0^{\circ}$&$45^{\circ}$&$0^{\circ}$\\
$D_{a,b}(|\nu_2{\rangle})$ &      $45^{\circ}$&$0^{\circ}$&$0^{\circ}$&$0^{\circ}$\\
$D_{a,b}(|\nu_3{\rangle})$ &      $0^{\circ}$&$0^{\circ}$&$0^{\circ}$&$0^{\circ}$\\
$D_{a,b}(|\nu_4{\rangle})$ &      $0^{\circ}$&$0^{\circ}$&$45^{\circ}$&$0^{\circ}$\\
$D_{a,b}(|\nu_5{\rangle})$ &      $22.5^{\circ}$&$0^{\circ}$&$45^{\circ}$&$0^{\circ}$\\
$D_{a,b}(|\nu_6{\rangle})$ &      $22.5^{\circ}$&$0^{\circ}$&$0^{\circ}$&$0^{\circ}$\\
$D_{a,b}(|\nu_7{\rangle})$ &      $22.5^{\circ}$&$45^{\circ}$&$0^{\circ}$&$0^{\circ}$\\
$D_{a,b}(|\nu_8{\rangle})$ &      $22.5^{\circ}$&$45^{\circ}$&$45^{\circ}$&$0^{\circ}$\\
$D_{a,b}(|\nu_9{\rangle})$ &      $22.5^{\circ}$&$45^{\circ}$&$22.5^{\circ}$&$0^{\circ}$\\
$D_{a,b}(|\nu_{10}{\rangle})$ &      $22.5^{\circ}$&$45^{\circ}$&$22.5^{\circ}$&$45^{\circ}$\\
$D_{a,b}(|\nu_{11}{\rangle})$ &      $22.5^{\circ}$&$0^{\circ}$&$22.5^{\circ}$&$45^{\circ}$\\
$D_{a,b}(|\nu_{12}{\rangle})$ &      $45^{\circ}$&$0^{\circ}$&$22.5^{\circ}$&$45^{\circ}$\\
$D_{a,b}(|\nu_{13}{\rangle})$ &      $0^{\circ}$&$0^{\circ}$&$22.5^{\circ}$&$45^{\circ}$\\
$D_{a,b}(|\nu_{14}{\rangle})$ &      $0^{\circ}$&$0^{\circ}$&$22.5^{\circ}$&$90^{\circ}$\\
$D_{a,b}(|\nu_{15}{\rangle})$ &      $45^{\circ}$&$0^{\circ}$&$22.5^{\circ}$&$90^{\circ}$\\
$D_{a,b}(|\nu_{16}{\rangle})$ &      $22.5^{\circ}$&$0^{\circ}$&$22.5^{\circ}$&$90^{\circ}$\\
   \hline
       \hline
  \end{tabular}
\end{table}

The matrixes $M_j$ ($1{\leq}j{\leq}16$) are

\begin{align}
&M_1=\frac{1}{2}\begin{bmatrix} 0 & 0 & 1 & 0 \\ 0 & 0 & -1-i & i \\ 1 & -1+i & 2 & -1-i \\ 0 & -i & -1+i & 0 \end{bmatrix} &
&M_2=\frac{1}{2}\begin{bmatrix} 0 & -1+i & 1 & 0 \\ -1-i & 2 & -1-i & i \\ 1 & -1+i & 0 & 0 \\ 0 & -i & 0 & 0 \end{bmatrix}\nonumber\\
&M_3=\frac{1}{2}\begin{bmatrix} 2 & -1+i & 1 & -1-i \\ -1-i & 0 & 0 & i \\ 1 & 0 & 0 & 0 \\ -1+i & -i & 0 & 0 \end{bmatrix} & &M_4=\frac{1}{2}\begin{bmatrix} 0 & 0 & 1 & -1-i \\ 0 & 0 & 0 & i \\ 1 & 0 & 0 & -1-i \\ -1+i & -i & -1+i & 2 \end{bmatrix}\nonumber\\
&M_5=\frac{1}{2}\begin{bmatrix} 0 & 0 & -1+i & 0 \\ 0 & 0 & 0 & 1-i \\ -1-i & 0 & 0 & 2i \\ 0 & 1+i & -2i & 0 \end{bmatrix} &
&M_6=\frac{1}{2}\begin{bmatrix} 0 & -2i & -1+i & 0 \\ 2i & 0 & 0 & 1-i \\ -1-i & 0 & 0 & 0 \\ 0 & 1+i & 0 & 0 \end{bmatrix}\nonumber\\
&M_7=\frac{1}{2}\begin{bmatrix} 0 & 2 & -1+i & 0 \\ 2 & 0 & 0 & -1+i \\ -1-i & 0 & 0 & 0 \\ 0 & -1-i & 0 & 0 \end{bmatrix} &
&M_8=\frac{1}{2}\begin{bmatrix} 0 & 0 & -1+i & 0 \\ 0 & 0 & 0 & -1+i \\ -1-i & 0 & 0 & 2 \\ 0 & -1-i & 2 & 0 \end{bmatrix}\nonumber\\
&M_9=\begin{bmatrix} 0 & 0 & -i & 0 \\ 0 & 0 & 0 & -i \\ i & 0 & 0 & 0 \\ 0 & i & 0 & 0 \end{bmatrix} &
&M_{10}=\begin{bmatrix} 0 & 0 & 1 & 0 \\ 0 & 0 & 0 & 1 \\ 1 & 0 & 0 & 0 \\ 0 & 1 & 0 & 0 \end{bmatrix}\nonumber
\end{align}
\begin{align}
&M_{11}=\begin{bmatrix} 0 & 0 & -i & 0 \\ 0 & 0 & 0 & i \\ i & 0 & 0 & 0 \\ 0 & -i & 0 & 0 \end{bmatrix} &
&M_{12}=\frac{1}{2}\begin{bmatrix} 0 & 0 & -1+i & 0 \\ 0 & 0 & 2 & -1-i \\ -1-i & 2 & 0 & 0 \\ 0 & -1+i & 0 & 0 \end{bmatrix}\nonumber\\
&M_{13}=\frac{1}{2}\begin{bmatrix} 0 & 0 & -1+i & 2 \\ 0 & 0 & 0 & -1-i \\ -1-i & 0 & 0 & 0 \\ 2 & -1+i & 0 & 0 \end{bmatrix} &
&M_{14}=\frac{1}{2}\begin{bmatrix} 0 & 0 & -1-i & 2i \\ 0 & 0 & 0 & 1-i \\ -1+i & 0 & 0 & 0 \\ -2i & 1+i & 0 & 0 \end{bmatrix}\nonumber\\
&M_{15}=\frac{1}{2}\begin{bmatrix} 0 & 0 & -1-i & 0 \\ 0 & 0 & 2i & 1-i \\ -1+i & -2i & 0 & 0 \\ 0 & 1+i & 0 & 0 \end{bmatrix} &
&M_{16}=\begin{bmatrix} 0 & 0 & 1 & 0 \\ 0 & 0 & 0 & -1 \\ 1 & 0 & 0 & 0 \\ 0 & -1 & 0 & 0 \end{bmatrix}
\end{align}

\ \\
\noindent{\bf The detection of classical entropy}
\ \\

We only need to record the distribution of click number of each detect-event while all rotation angles of HWPs and QWPs are $0^{\circ}$.

\newpage

\subsection{D. Counter-examples for the hypotheses in Discussion of the main text}


\begin{hypothesis}
$\min_{\rho_x\in\mathbb{C}^d}S(\rho)=\min_{\rho_x\in\mathbb{C}^n}S(\rho)$ while $\sum_{x=1}^n\sum_{y=1}^l\alpha_{xy}\tr(\rho_xM_y)=w_d$, $\rho=\sum_{x=1}^n\rho_x/n$, and $L_{d-1}^{(q)}{<}w_d{\leq}L_d^{(q)}$, ($d{<}n$) where $L_d^{(q)}$ is the $d$-dimensional quantum bound of the dimension witness $w_d$.
\end{hypothesis}


\noindent Counter-example 1:

From Eq.(3) of the main text, the dimension witness $I_4$ can be written as
\begin{align}
I_4=&\tr[\rho_1(M_1+M_2+M_3)]+\tr[\rho_2(M_1+M_2-M_3)]+\tr[\rho_3(M_1-M_2)]+\tr[\rho_1(-M_1)]\nonumber\\
{\leq}&\lambda_{\max}(M_1+M_2+M_3)+\lambda_{\max}(M_1+M_2-M_3)+\lambda_{\max}(M_1-M_2)+\lambda_{\max}(-M_1) \label{TT_0}
\end{align}
where $\lambda_{\max}(\Omega)$ represents the maximum eigenvalue of observable $\Omega$.

Let $M_k=2\hat{U}^{-1}|m_k{\rangle}{\langle}m_k|\hat{U}-1$  where $\hat{U}$ is a 2 order unitary matrix and $1{\leq}k{\leq}3$. Since $\{|m_k{\rangle}\}$ are in $\mathbb{C}^2$, without loss of generality, let
\begin{align}
|m_1{\rangle}=&(1,0)\\
|m_2{\rangle}=&(\cos{\frac{\theta}{2}},\sin{\frac{\theta}{2}}) \text{ , while $\theta\in[0,\pi]$ } \label{TT_6}\\
|m_3{\rangle}=&(\cos{\frac{\phi}{2}},\sin{\frac{\phi}{2}}e^{i\varphi}) \text{ , while $\phi\in(-\pi,\pi]$ and $\varphi\in[0,\pi)$ }
\end{align}

Then
\begin{align}
M_1+M_2+M_3=&\hat{U}^{-1}\begin{bmatrix} 1+\cos\theta+\cos\phi & \sin\theta+\sin\phi{e^{i\varphi}} \\ \sin\theta+\sin\phi{e^{-i\varphi}} & -1-\cos\theta-\cos\phi \end{bmatrix}\hat{U} \label{QQ_1}\\
M_1+M_2-M_3=&\hat{U}^{-1}\begin{bmatrix} 1+\cos\theta-\cos\phi & \sin\theta-\sin\phi{e^{i\varphi}} \\ \sin\theta-\sin\phi{e^{-i\varphi}} & -1-\cos\theta+\cos\phi \end{bmatrix}\hat{U}\\
M_1-M_2=&\hat{U}^{-1}\begin{bmatrix} 1-\cos\theta & -\sin\theta \\ -\sin\theta & -1+\cos\theta \end{bmatrix}\hat{U}\\
-M_1=&\hat{U}^{-1}\begin{bmatrix} -1 & 0 \\ 0 & 1 \end{bmatrix}\hat{U} \label{QQ_2}
\end{align}

Hence,
\begin{align}
\lambda_{\max}(M_1+M_2+M_3)=&\sqrt{(3+2\cos{\theta})+(2\cos\theta\cos\phi+2\sin\theta\sin\phi\cos\varphi)}\label{TT_1}\\
\lambda_{\max}(M_1+M_2-M_3)=&\sqrt{(3+2\cos{\theta})-(2\cos\theta\cos\phi+2\sin\theta\sin\phi\cos\varphi)}\\
\lambda_{\max}(M_1-M_2)=&\sqrt{2-2\cos{\theta}}\\
\lambda_{\max}(-M_1)=&1 \label{TT_2}
\end{align}

Substitute Eq.(\ref{TT_1})$\sim$(\ref{TT_2}) into Eq.(\ref{TT_0}),
\begin{align}
I_4{\leq}&\sqrt{(3+2\cos{\theta})+(2\cos\theta\cos\phi+2\sin\theta\sin\phi\cos\varphi)}+\nonumber\\
&\sqrt{(3+2\cos{\theta})-(2\cos\theta\cos\phi+2\sin\theta\sin\phi\cos\varphi)}+\nonumber\\
&\sqrt{2-2\cos{\theta}}+1\nonumber\\
{\leq}&2\sqrt{(3+2\cos{\theta})}+\sqrt{2-2\cos{\theta}}+1\nonumber\\
=&\frac{1}{2}\sqrt{(3+2\cos{\theta})}+\frac{1}{2}\sqrt{(3+2\cos{\theta})}+\frac{1}{2}\sqrt{(3+2\cos{\theta})}+\frac{1}{2}\sqrt{(3+2\cos{\theta})}+\sqrt{2-2\cos{\theta}}+1\nonumber\\
{\leq}&\sqrt{5[\frac{3+2\cos{\theta}}{4}+\frac{3+2\cos{\theta}}{4}+\frac{3+2\cos{\theta}}{4}+\frac{3+2\cos{\theta}}{4}+(2-2\cos{\theta})]}+1\nonumber\\
=&6 \label{TT_4}
\end{align}

The second sign of less than or equal to ($\leq$) becomes equal to (=) if
\begin{align}
2\cos\theta\cos\phi+2\sin\theta\sin\phi\cos\varphi=0 \label{TT_3}
\end{align}

The third sign of less than or equal to ($\leq$) becomes equal to (=) if
\begin{align}
\frac{1}{2}\sqrt{(3+2\cos{\theta})}=&\sqrt{2-2\cos{\theta}}\nonumber\\
\Rightarrow\cos\theta=&0.5 \label{TT_5}
\end{align}

Considering Eq.(\ref{TT_6}), from Eq.(\ref{TT_5}) we can obtain
\begin{align}
\theta=\frac{\pi}{3} \label{TT_7}
\end{align}

On the other hand, the vectors $\{|v_k{\rangle}\}$, where $\rho_k=\hat{U}^{-1}|v_k{\rangle}{\langle}v_k|\hat{U}$, should be the eigenvectors corresponding the maximum eigenvalues of Eq.(\ref{QQ_1})$\sim$(\ref{QQ_2}). Hence,
\begin{align}
|v_1{\rangle}=&\frac{(1+\cos\theta+\cos\phi+\sqrt{(1+\cos\theta+\cos\phi)^2+|\sin\theta+\sin\phi{e^{i\varphi}}|^2},\sin\theta+\sin\phi{e^{-i\varphi}})}{\sqrt{2[{(1+\cos\theta+\cos\phi)^2+|\sin\theta+\sin\phi{e^{i\varphi}}|^2}]+2(1+\cos\theta+\cos\phi)\sqrt{(1+\cos\theta+\cos\phi)^2+|\sin\theta+\sin\phi{e^{i\varphi}}|^2}}} \label{TT_8}\\
|v_2{\rangle}=&\frac{(1+\cos\theta-\cos\phi+\sqrt{(1+\cos\theta-\cos\phi)^2+|\sin\theta-\sin\phi{e^{i\varphi}}|^2},\sin\theta-\sin\phi{e^{-i\varphi}})}{\sqrt{2[{(1+\cos\theta-\cos\phi)^2+|\sin\theta-\sin\phi{e^{i\varphi}}|^2}]+2(1+\cos\theta-\cos\phi)\sqrt{(1+\cos\theta-\cos\phi)^2+|\sin\theta-\sin\phi{e^{i\varphi}}|^2}}}\\
|v_3{\rangle}=&\frac{(1-\cos\theta+\sqrt{(1-\cos\theta)^2+|\sin\theta|^2},-\sin\theta)}{\sqrt{2[{(1-\cos\theta)^2+|\sin\theta|^2}]+2(1-\cos\theta)\sqrt{(1-\cos\theta)^2+|\sin\theta|^2}}}\\
|v_3{\rangle}=&(0,1) \label{TT_9}
\end{align}

Substitute Eq.(\ref{TT_3}) and Eq.(\ref{TT_7}) into Eq.(\ref{TT_8})$\sim$(\ref{TT_9}),
\begin{align}
\rho_1=\hat{U}^{-1}|v_1{\rangle}{\langle}v_1|\hat{U}=&\hat{U}^{-1}\frac{1}{8}\begin{bmatrix} 7+2{\cos\phi} & \sqrt{3}+2{\sin{\phi}e^{-i\varphi}} \\ {\sqrt{3}}+2{\sin{\phi}e^{i\varphi}} & {1}-2{\cos\phi} \end{bmatrix}\hat{U} \label{TT_10}\\
\rho_2=\hat{U}^{-1}|v_2{\rangle}{\langle}v_2|\hat{U}=&\hat{U}^{-1}\frac{1}{8}\begin{bmatrix} {7}-2{\cos\phi} & {\sqrt{3}}-2{\sin{\phi}e^{-i\varphi}} \\ {\sqrt{3}}-2{\sin{\phi}e^{i\varphi}} & {1}+2{\cos\phi} \end{bmatrix}\hat{U} \\
\rho_3=\hat{U}^{-1}|v_3{\rangle}{\langle}v_3|\hat{U}=&\hat{U}^{-1}\frac{1}{4}\begin{bmatrix} 3 & -{\sqrt{3}} \\ -{\sqrt{3}} & 1 \end{bmatrix}\hat{U} \\
\rho_4=\hat{U}^{-1}|v_4{\rangle}{\langle}v_4|\hat{U}=&\hat{U}^{-1}\begin{bmatrix} 0 & 0 \\ 0 & 1 \end{bmatrix}\hat{U}  \label{TT_11}
\end{align}

Hence,

\begin{align}
\rho=\frac{1}{4}(\rho_1+\rho_2+\rho_3+\rho_4)=\hat{U}^{-1}\frac{1}{8}\begin{bmatrix} 5 & 0 \\ 0 & 3 \end{bmatrix}\hat{U}
\end{align}

For $I_4=L_2^{(q)}=6$, {{although $\phi$, $\varphi$ and $\hat{U}$ are not unique, the Von Neumann entropy of $\rho$ is unique}}. Then while $I_4=6$,
\begin{align}
\min_{\rho_x\in\mathbb{C}^2}S(\rho)=0.954 \text{ bit} \label{I_4-counter example-1}
\end{align}

On the other hand, there exist states for ququart
\begin{align}
|\psi_1{\rangle}=&(1,0,0,0)\\
|\psi_2{\rangle}=&(0.8290,0.5592,0,0)\\
|\psi_3{\rangle}=&(0.7660,-0.6428,0,0)\\
|\psi_4{\rangle}=&(0.8844,-0.0191,-0.1204,0.4506)
\end{align}
and the measurement operators $M_y=1-2|m_y{\rangle}{\langle}m_y|$
\begin{align}
|m_1{\rangle}=&(0.2229,-0.0058,-0.2516,0.9418)\\
|m_2{\rangle}=&(0.4838,-0.8752,0,0)\\
|m_3{\rangle}=&(0.4695,0.8829,0,0)
\end{align}
where $I_4=6.000$ and $S(\rho)=0.912$ bit. Hence,
\begin{align}
\min_{\rho_x\in\mathbb{C}^4}S(\rho){\leq}0.9122 \text{ bit} \label{I_4-counter example-2}
\end{align}
From Eq.(\ref{I_4-counter example-1}) and Eq.(\ref{I_4-counter example-2}),
\begin{align}
\min_{\rho_x\in\mathbb{C}^2}S(\rho)>\min\limits_{\rho_x\in\mathbb{C}^4}S(\rho)
\end{align}
which disproves the hypothesis.\\


\noindent Counter-example 2:

For $R_4=L_3^{(q)}=6.472$, there are sets of states $\rho_x$ in $\mathbb{C}^3$\cite{ahrens2014prl}. The states are $\rho_x=\hat{U}^{-1}|\psi_x{\rangle}{\langle}\psi_x|\hat{U}$ where $\hat{U}$ is a 3 order unitary matrix and
\begin{align}
|\psi_1{\rangle}=&(0,0,1)\\
|\psi_2{\rangle}=&(\frac{1}{\sqrt{2}},-\frac{1}{\sqrt{2}},0)\\
|\psi_3{\rangle}=&(\frac{1}{\sqrt{2}},\frac{1}{\sqrt{2}},0)\\
|\psi_4{\rangle}=&(1,0,0)
\end{align}
The measurement operators are $M_y=1-2\hat{U}^{-1}|m_y{\rangle}{\langle}m_y|\hat{U}$ where
\begin{align}
|m_1{\rangle}=&(\frac{\sqrt{5}+1}{\sqrt{10+2\sqrt{5}}},\frac{2}{\sqrt{10+2\sqrt{5}}},0)\\
|m_2{\rangle}=&(\frac{\sqrt{5}+1}{\sqrt{10+2\sqrt{5}}},\frac{-2}{\sqrt{10+2\sqrt{5}}},0)
\end{align}

Although $\hat{U}$ is not unique, the Von Neumann entropy of $\rho$ is unique. Then while $R_4=6.472$,
\begin{align}
\min_{\rho_x\in\mathbb{C}^3}S(\rho)=1.5 \text{ bit} \label{R_4-counter example-1}
\end{align}

On the other hand, there exist states for ququart
\begin{align}
|\psi_1{\rangle}=&(1,0,0,0)\\
|\psi_2{\rangle}=&(0.5892,0.5736,0.5690,0)\\
|\psi_3{\rangle}=&(-0.6257,0.5584,0.0293,0.5439)\\
|\psi_4{\rangle}=&(0.0175,0.9998,0,0)
\end{align}
and the measurement operators $M_y=1-2|m_y^{(1)}{\rangle}{\langle}m_y^{(1)}|-2|m_y^{(2)}{\rangle}{\langle}m_y^{(2)}|$
\begin{align}
|m_1^{(1)}{\rangle}=&(-0.2925,0.8860,-0.0987,0.3460)\\
|m_1^{(2)}{\rangle}=&(-0.1432,-0.3525,0.3117,0.8707)\\
|m_2^{(1)}{\rangle}=&(0.2906,0.8847,0.3496,-0.1030)\\
|m_2^{(2)}{\rangle}=&(0.1143,-0.3604,0.8911,0.2511)
\end{align}
where $R_4=6.472$ and $S(\rho)=1.418$ bit. Hence,
\begin{align}
\min_{\rho_x\in\mathbb{C}^4}S(\rho){\leq}1.418 \text{ bit} \label{R_4-counter example-2}
\end{align}

From Eq.(\ref{R_4-counter example-1}) and Eq.(\ref{R_4-counter example-2}),
\begin{align}
\min\limits_{\rho_x\in\mathbb{C}^3}S(\rho)>\min\limits_{\rho_x\in\mathbb{C}^4}S(\rho)
\end{align}
which also disproves the hypothesis. This is also shown in the FIG.~3 in the Supplementary Material of Ref.\cite{chaves2015prl}.\\


\begin{hypothesis}
For any linear dimension witness $\sum_{x=1}^n\sum_{y=1}^l\alpha_{xy}\tr(\rho_xM_y)=w_d$, the right part of Eq.(11) in Ref.\cite{chaves2015prl} is the minimal classical entropy.
\end{hypothesis}


Let $\lambda_{i,j}$ be the $j$th strategy for an $i$-dimensional classical system. $w(\lambda_{i,j})$ represents the classical dimension witness,
\begin{align}
w(\lambda_{i,j})=\sum_{x=1}^{n}\sum_{y=1}^{l}\alpha_{xy}E_{xy}^{(\lambda_{i,j})}
\end{align}

The maximal value of the dimension witness for the $d$-dimensional classical system is $L_d$, hence
\begin{align}
L_d=\max_{j}w(\lambda_{d,j})\label{maxx}
\end{align}

Here, we use $H([p_0,p_1,\ldots,p_{n-2},p_{n-1}])$ to represent the classical entropy $H(M)$. Without loss of generality, let $p_0{\geq}p_1{\geq}\ldots{\geq}p_{n-2}{\geq}p_{n-1}$. Then
\begin{align}
H([p_0,p_1,\ldots,p_{n-2},p_{n-1}])=H(M)=\sum_{k=0}^{n-1}-p_k\log_2p_k \label{AA}
\end{align}
For the case of the $d$-dimensional system where $d<n$, $p_k=0$ while $d{\leq}k{\leq}n-1$. Then we use $\lim_{x\rightarrow0}x\log_2x=0$ to keep the effectivity of the Eq.(\ref{AA}).

\ \\
\noindent Counter-example 3 :

Let
\begin{align}
A_{xy}=\begin{bmatrix} \alpha_{x=1,y=1} & \alpha_{x=1,y=2} \\ \alpha_{x=2,y=1} & \alpha_{x=2,y=2} \\ \alpha_{x=3,y=1} & \alpha_{x=3,y=2} \\ \alpha_{x=4,y=1} & \alpha_{x=4,y=2} \end{bmatrix}=\begin{bmatrix} 0.4955&0.7775\\ -0.6092&-0.6572\\ 0.0048&-0.5283\\ -0.5877&0.8258 \end{bmatrix}
\end{align}

Then the maximal value of the dimension witness by the 4-dimensional classical system is $L_4=4.4860$, while
\begin{align}
E_{m,y}^{(\lambda_{4,1})}=\begin{bmatrix}   E(m=0,y=1,\lambda_{4,1}) &  \ldots &  E(m=3,y=1,\lambda_{4,1}) \\ E(m=0,y=2,\lambda_{4,1}) &  \ldots &  E(m=3,y=2,\lambda_{4,1}) \end{bmatrix}=\begin{bmatrix}   1&-1&1&-1 \\ 1&-1&-1&1 \end{bmatrix}
\end{align}
and
\begin{align}
P_{m,x}^{(\lambda_{4,1})}=\begin{bmatrix}   P(m=0|x=1,\lambda_{4,1}) & \ldots &  P(m=0|x=4,\lambda_{4,1}) \\ & \vdots & \\   P(m=3|x=1,\lambda_{4,1}) & \ldots &  P(m=3|x=4,\lambda_{4,1}) \end{bmatrix}=\begin{bmatrix} 1&0&0&0 \\ 0&1&0&0 \\ 0&0&1&0 \\ 0&0&0&1 \end{bmatrix}
\end{align}

The maximal value of the dimension witness by the 3-dimensional classical system is $L_3=4.4764$, while
\begin{align}
E_{m,y}^{(\lambda_{3,1})}=\begin{bmatrix}   E(m=0,y=1,\lambda_{3,1}) &  \ldots &  E(m=3,y=1,\lambda_{3,1}) \\ E(m=0,y=2,\lambda_{3,1}) &  \ldots &  E(m=3,y=2,\lambda_{3,1}) \end{bmatrix}=\begin{bmatrix}   -1&1&-1&1 \\ -1&1&1&-1 \end{bmatrix}
\end{align}

and
\begin{align}
P_{m,x}^{(\lambda_{3,1})}=\begin{bmatrix}   P(m=0|x=1,\lambda_{3,1}) & \ldots &  P(m=0|x=4,\lambda_{3,1}) \\ & \vdots & \\   P(m=3|x=1,\lambda_{3,1}) & \ldots &  P(m=3|x=4,\lambda_{3,1}) \end{bmatrix}=\begin{bmatrix} 0&1&1&0 \\ 1&0&0&0 \\ 0&0&0&1 \\ 0&0&0&0 \end{bmatrix}
\end{align}

The maximal value of the dimension witness by the 2-dimensional classical system is $L_2^{(c)}=3.4854$, while
\begin{align}
E_{m,y}^{(\lambda_{2,1})}=\begin{bmatrix}   E(m=0,y=1,\lambda_{2,1}) &  \ldots &  E(m=3,y=1,\lambda_{2,1}) \\ E(m=0,y=2,\lambda_{2,1}) &  \ldots &  E(m=3,y=2,\lambda_{2,1}) \end{bmatrix}=\begin{bmatrix}   1&-1&1&-1 \\ 1&-1&-1&1 \end{bmatrix}
\end{align}

and
\begin{align}
P_{m,x}^{(\lambda_{2,1})}=\begin{bmatrix}   P(m=0|x=1,\lambda_{2,1}) & \ldots &  P(m=0|x=4,\lambda_{2,1}) \\ & \vdots & \\   P(m=3|x=1,\lambda_{2,1}) & \ldots &  P(m=3|x=4,\lambda_{2,1}) \end{bmatrix}=\begin{bmatrix} 1&0&0&1 \\ 0&1&1&0 \\ 0&0&0&0 \\ 0&0&0&0 \end{bmatrix}\label{TTTTT}
\end{align}

The maximal value of the dimension witness by the 1-dimensional classical system is $L_1=1.1144$, while
\begin{align}
E_{m,y}^{(\lambda_{1,1})}=\begin{bmatrix}   E(m=0,y=1,\lambda_{1,1}) &  \ldots &  E(m=3,y=1,\lambda_{1,1}) \\ E(m=0,y=2,\lambda_{1,1}) &  \ldots &  E(m=3,y=2,\lambda_{1,1}) \end{bmatrix}=\begin{bmatrix}   -1&-1&1&1 \\ 1&-1&-1&1 \end{bmatrix}
\end{align}

and
\begin{align}
P_{m,x}^{(\lambda_{1,1})}=\begin{bmatrix}   P(m=0|x=1,\lambda_{1,1}) & \ldots &  P(m=0|x=4,\lambda_{1,1}) \\ & \vdots & \\   P(m=3|x=1,\lambda_{1,1}) & \ldots &  P(m=3|x=4,\lambda_{1,1}) \end{bmatrix}=\begin{bmatrix} 1&1&1&1 \\ 0&0&0&0 \\ 0&0&0&0 \\ 0&0&0&0 \end{bmatrix}
\end{align}

For dimension witness $w_d=L_2$, the minimal classical entropy is $H([p_0,\ldots,p_3])$, subject to
\begin{align}
\sum_{x=1}^4\sum_{i=1}^4\sum_{j}q_{\lambda_{i,j}}P(m=k|x,\lambda_{i,j})/4&=p_k \text{ while }0{\leq}k{\leq}3\\
\sum_{i=1}^4\sum_{j}q_{\lambda_{i,j}}w(\lambda_{i,j})&=L_2 \label{sup-last-1}\\
\sum_{i=1}^4\sum_{j}q_{\lambda_{i,j}}&=1 \label{sup-last-2}
\end{align}

Following Eq.(\ref{sup-last-1}), Eq.(\ref{sup-last-2}), and Eq.(\ref{maxx})
\begin{align}
L_2&{\leq}\sum_{i=1}^{4}\sum_{j}q_{\lambda_{i,j}}\max_{j}w(\lambda_{i,j})\nonumber\\
&=\sum_{i=1}^{4}\sum_{j}q_{\lambda_{i,j}}L_i\nonumber\\
&=\sum_{j}q_{\lambda_{1,j}}L_1+(1-\sum_{j}q_{\lambda_{1,j}}-\sum_{j}q_{\lambda_{3,j}}-\sum_{j}q_{\lambda_{4,j}})L_2+\sum_{j}q_{\lambda_{3,j}}L_3+\sum_{j}q_{\lambda_{4,j}}L_4\nonumber\\
\Rightarrow\sum_jq_{\lambda_{1,j}}&{\leq}\frac{L_3-L_2}{L_2-L_1}\sum_jq_{\lambda_{3,j}}+\frac{L_4-L_2}{L_2-L_1}\sum_jq_{\lambda_{4,j}}\nonumber\\
&=0.4180\sum_jq_{\lambda_{3,j}}+0.4220\sum_jq_{\lambda_{4,j}} \label{sup-last-3}
\end{align}

Then, for $p_0$,
\begin{align}
p_0&=\sum_{x=1}^4\sum_{\lambda}q_{\lambda}P(m=0|x,\lambda)/4\nonumber\\
&=\sum_{i=1}^{4}\sum_{j}q_{\lambda_{i,j}}\{\sum_{x=1}^4P(m=0|x,\lambda_{i,j})/4\} \label{sup-last-4}
\end{align}

Considering that for different dimensional systems, $\sum_{x=1}^4P(m=0|x,\lambda_{i,j})/4$ have different upper bounds. $\sum_{x=1}^4P(m=0|x,\lambda_{1,j})/4{\leq}1$ for 1-dimensional systems, $\sum_{x=0}^4P(m=0|x,\lambda_{2,j})/4{\leq}3/4$ for 2-dimensional systems, $\sum_{x=1}^4P(m=0|x,\lambda_{3,j})/4{\leq}1/2$ for 3-dimensional systems, and $\sum_{x=1}^4P(m=0|x,\lambda_{4,j})/4{\leq}1/4$ for 4-dimensional systems. Here we notice that $\sum_{x=1}^4P(m=0|x,\lambda_{2,1})/4{\leq}1/2$ for the case of $\lambda_{2,1}$ from Eq.(\ref{TTTTT}). Hence
\begin{align}
p_0&{\leq}\sum_{j}q_{\lambda_{1,j}}\cdot1+\sum_{j{\neq}2}q_{\lambda_{2,j}}\cdot\frac{3}{4}+q_{\lambda_{2,1}}\cdot\frac{1}{2}+\sum_{j}q_{\lambda_{3,j}}\cdot\frac{1}{2}+\sum_{j}q_{\lambda_{4,j}}\cdot\frac{1}{4}\nonumber\\
&=\frac{3}{4}-\frac{1}{4}q_{\lambda_{2,1}}+\frac{1}{4}[\sum_{j}q_{\lambda_{1,j}}-\sum_{j}q_{\lambda_{3,j}}-2\sum_{j}q_{\lambda_{4,j}}]\nonumber\\
&=\frac{3}{4}-\frac{1}{4}q_{\lambda_{2,1}}+\frac{1}{4}[\sum_{j}q_{\lambda_{1,j}}-0.4180\sum_{j}q_{\lambda_{3,j}}-0.4220\sum_{j}q_{\lambda_{4,j}}]-\frac{0.5820}{4}\sum_{j}q_{\lambda_{3,j}}-\frac{1.5780}{4}\sum_{j}q_{\lambda_{4,j}}\nonumber\\
&{\leq}\frac{3}{4}-\frac{1}{4}q_{\lambda_{2,1}}-\frac{0.5820}{4}\sum_{j}q_{\lambda_{3,j}}-\frac{1.5780}{4}\sum_{j}q_{\lambda_{4,j}}\label{sup-last-5}
\end{align}

Since Eq.(\ref{sup-last-1}), $q_{\lambda_{2,1}}$, $\sum_{j}q_{\lambda_{3,j}}$ and $\sum_{j}q_{\lambda_{4,j}}$ can't be 0 simultaneously,
\begin{align}
p_0<\frac{3}{4}\label{p0upper}
\end{align}

Since $p_0{\geq}p_1{\geq}p_2{\geq}p_3$ and $w_d=L_2>L_1$,
\begin{align}
p_0>\frac{1}{4}\label{p0lower}
\end{align}

Since $-x\log_2x-y\log_2y{\geq}-(x+y)\log_2(x+y)$ while $x{\geq}0$, $y{\geq}0$ and $x+y{\leq}1$, then
\begin{align}
H([p_0,p_1,p_2,p_3]){\geq}H([p_0,p_1+p_2+p_3,0,0])
\end{align}

Since $-x\log_2x-(1-x)\log_2(1-x)>-y\log_2y-(1-y)\log_2(1-y)$ while $x{\geq}0$, $y{\geq}0$ and $0<x<y{\leq}\frac{1}{2}$, then considering Eq.(\ref{p0upper}) and Eq.(\ref{p0lower}),
\begin{align}
H([p_0,p_1+p_2+p_3,0,0])>H([\frac{3}{4},\frac{1}{4},0,0])=-\frac{3}{4}\log_2\frac{3}{4}-\frac{1}{4}\log_2\frac{1}{4}=0.811 \text{ bit}
\end{align}
Hence
\begin{align}
{H([p_0,p_1,p_2,p_3])}{>}0.811 \text{ bit}
\end{align}

On the other hand, while using the strategy of Eq.(11) in Ref.\cite{chaves2015prl}, for the case of $w_d=L_2^{(c)}=3.4854$, the minimal classical entropy $H(M)$ is $0.811$ bit. Hence, the hypothesis is disproved.





\begin{thebibliography}{99}

\bibitem{acin2007prl}
A. Ac\'{i}n, N. Brunner, N. Gisin, S. Massar, S. Pironio, and V. Scarani, Phys. Rev. Lett. {\bf98}, 230501 (2007).

\bibitem{Elkert1991prl}
A. K. Ekert, Phys. Rev. Lett. {\bf67}, 661 (1991).



\bibitem{liang2011prl}
J.-D. Bancal, N. Gisin, Y.-C. Liang, and S. Pironio, Phys. Rev. Lett. {\bf106}, 250404 (2011).

\bibitem{Rabelo2011prl}
R. Rabelo, M. Ho, D. Cavalcanti, N. Brunner, and V. Scarani, Phys. Rev. Lett. {\bf107}, 050502 (2011).

\bibitem{liang2013prl}
T. Moroder, J.-D. Bancal, Y.-C. Liang, M. Hofmann, and O. G\"{u}hne, Phys. Rev. Lett. {\bf111}, 030501 (2013).

\bibitem{Barreiro2013natphy}
J. T. Barreiro, J.-D. Bancal, P. Schindler, D. Nigg, M. Hennrich, T. Monz, N. Gisin, and R. Blatt, Nat. Phys. {\bf9}, 559 (2013).


\bibitem{brunner2008prl}
N. Brunner, S. Pironio, A. Acin, N. Gisin, A. A. M\'{e}thot, and V. Scarani, Phys. Rev. Lett. {\bf100}, 210503 (2008).

\bibitem{gallego2010prl}
R. Gallego, N. Brunner, C. Hadley, and A. Ac\'{i}n, Phys. Rev. Lett. {\bf105}, 230501 (2010).

\bibitem{bowles2014prl}
J. Bowles, M. T. Quintino, and N. Brunner, Phys. Rev. Lett. {\bf112}, 140407 (2014).



\bibitem{pawlowski2011pra}
M. Paw{\l}owski and N. Brunner, Phys. Rev. A {\bf84}, 010302 (2011).




\bibitem{hendrych2014natphy}
M. Hendrych, R. Gallego, M. Mi\u{c}uda, N. Brunner, A. Ac\'{i}n, and J. P. Torres, Nat. Phys. {\bf8}, 588 (2012).

\bibitem{ahrens2012natphy}
J. Ahrens, P. Badzi\c{a}g, A. Cabello, and M. Bourennane, Nat. Phys. {\bf8}, 592 (2012).

\bibitem{ambrosio2014prl}
V. D'Ambrosio, F. Bisesto, F. Sciarrino, J. F. Barra, G. Lima, and A. Cabello, Phys. Rev. Lett. {\bf112}, 140503 (2014).

\bibitem{ahrens2014prl}
J. Ahrens, P. Badzi\c{a}g, M. Paw{\l}owski, M. \.{Z}ukowski, and M. Bourennane, Phys. Rev. Lett. {\bf112}, 140401 (2014).





\bibitem{holevo}
A. Holevo, Probl. Inf. Transm. {\bf9}, 177 (1973).

\bibitem{palowski2009nat}
M. Paw{\l}owski, T. Paterek, D. Kaszlikowski, V. Scarani, A. Winter, and M. \.{Z}ukowski, Nature (London) {\bf461}, 1101 (2009).

\bibitem{chaves2015prl}
R. Chaves, J. B. Brask, and N. Brunner, Phys. Rev. Lett. {\bf115}, 110501 (2015).

\bibitem{pearl2009}
J. Pearl, \textit{Causality} (Cambridge University Press, Cambridge, England, 2009).

\bibitem{chaves2015nc}
R. Chaves, C. Majenz, and D. Gross, Nat. Commun. {\bf6}, 5766 (2015).

\bibitem{chaves2013arxiv}
R. Chaves, L. Luft, T. O. Maciel, D. Gross, D. Janzing, and B. Sch\"{o}lkopf, \textit{Proceedings of the 30th Conference on Uncertainty in Artificial Intelligence}, (AUAI Press Corvallis, Oregon, 2014), p. 112.

\bibitem{chaves2013ieee}
R. Chaves and T. Fritz, Phys. Rev. A {\bf85}, 032113 (2012); T. Fritz and R. Chaves, IEEE Trans. Inf. Theory {\bf59}, 803 (2013); R. Chaves, L. Luft, and D. Gross, New J. Phys. {\bf16}, 043001 (2014).







\bibitem{onto}
N. Harrigan, T. Rudolph, and S. Aaronson, arXiv:0709.1149; E. F. Galvao, Phys. Rev. A {\bf80}, 022106 (2009).




\bibitem{bogdanov2004prl}
Y. I. Bogdanov, M. V. Chekhova, S. P. Kulik, G. A. Maslennikov, A. A. Zhukov, C. H. Oh, and M. K. Tey, Phys. Rev. Lett. {\bf93}, 230503 (2004).

\bibitem{james2001pra}
D. F. V. James, P. G. Kwiat, W. J. Munro, and A. G. White, Phys. Rev. A {\bf64}, 052312 (2001).

\bibitem{thew2002pra}
R. T. Thew, K. Nemoto, A. G. White, and W. J. Munro, Phys. Rev. A {\bf66}, 012303 (2002).

\bibitem{bogdanov2004pra}
Yu. I. Bogdanov, M. V. Chekhova, L. A. Krivitsky, S. P. Kulik, A. N. Penin, A. A. Zhukov, L. C. Kwek, C. H. Oh, and M. K. Tey, Phys. Rev. A {\bf70}, 042303 (2004).


\bibitem{NTT}
H. Takesue, and K. Inoue, Physical Review A {\bf70}, 031802 (2004).

\bibitem{kumar}
X. Li, P. L. Voss, J. E. Sharping, and P. Kumar, Phys. Rev. Lett. {\bf94}, 053601 (2005).

\bibitem{fiber}
Q. Zhou, W. Zhang, J. Cheng, Y. Huang, and J. Peng, Opt. Lett. {\bf34}(18), 2706-2708 (2009).





\bibitem{Wehrl1978}
A. Wehrl, Rev. Mod. Phys. {\bf50}, 221 (1978).





\end{thebibliography}
\end{document}